\documentclass{article}

\usepackage{amsmath,amssymb,amsthm}
\usepackage{fullpage}
\usepackage{xcolor,xspace}
\usepackage{graphicx}
\usepackage{caption}
\usepackage{booktabs}
\usepackage{microtype}
\usepackage{enumitem}
\usepackage{thmtools}
\usepackage{thm-restate}
\usepackage{algorithm}
\usepackage[
    indLines=true,
    noEnd=true,
    rightComments=true,
    italicComments=true,
]{algpseudocodex}
\algrenewcommand\algorithmicrequire{\textbf{Input:}}
\algrenewcommand\algorithmicensure{\textbf{Output:}}

\usepackage{pgfplots}
\pgfplotsset{compat=1.18}

\declaretheorem[name=Theorem]{theorem}
\declaretheorem[name=Lemma,sibling=theorem]{lemma}
\declaretheorem[name=Corollary,sibling=theorem]{corollary}

\usepackage[
    colorlinks=true,
    linkcolor=blue!62!black,
    citecolor=green!48!black,
    urlcolor=blue!70!black,
    linktoc=page
]{hyperref}
\usepackage[nameinlink,noabbrev]{cleveref}

\crefname{theorem}{Theorem}{Theorems}
\Crefname{theorem}{Theorem}{Theorems}
\crefname{lemma}{Lemma}{Lemmas}
\Crefname{lemma}{Lemma}{Lemmas}
\crefname{corollary}{Corollary}{Corollaries}
\Crefname{corollary}{Corollary}{Corollaries}
\crefname{observation}{Observation}{Observations}
\Crefname{observation}{Observation}{Observations}
\crefname{proposition}{Proposition}{Propositions}
\Crefname{proposition}{Proposition}{Propositions}
\crefname{claim}{Claim}{Claims}
\Crefname{claim}{Claim}{Claims}
\crefname{conjecture}{Conjecture}{Conjectures}
\Crefname{conjecture}{Conjecture}{Conjectures}
\crefname{assumption}{Assumption}{Assumptions}
\Crefname{assumption}{Assumption}{Assumptions}
\crefname{definition}{Definition}{Definitions}
\Crefname{definition}{Definition}{Definitions}
\crefname{remark}{Remark}{Remarks}
\Crefname{remark}{Remark}{Remarks}
\crefname{algorithm}{Algorithm}{Algorithms}
\Crefname{algorithm}{Algorithm}{Algorithms}
\crefname{figure}{Figure}{Figures}
\Crefname{figure}{Figure}{Figures}
\crefname{section}{Section}{Sections}
\Crefname{section}{Section}{Sections}
\crefname{appendix}{Appendix}{Appendices}
\Crefname{appendix}{Appendix}{Appendices}

\tikzset{
  algpxIndentLine/.style={draw=black!100,very thin}
}

\newcommand{\R}{\mathbb{R}}

\newcommand{\Z}{\mathbb{Z}}

\newcommand{\softO}{\widetilde O}

\newcommand{\wt}{\operatorname{wt}}
\newcommand{\prf}{\operatorname{pr}}
\newcommand{\SubSums}[1]{\mathcal{S}(#1)}
\newcommand{\mc}{\mathbin{\star_{\min}}}
\newcommand{\xc}{\mathbin{\star_{\max}}}

\allowdisplaybreaks

\makeatletter
\providecommand{\theHALG@line}{}
\renewcommand{\theHALG@line}{\thealgorithm.\arabic{ALG@line}}
\makeatother

\title{Simpler Algorithms for Knapsack,\\
Subset Sum, and Min-Plus Convolution}
\author{Trevor Vaughn\thanks{Carnegie Mellon University. Email: tnvaughn@cmu.edu}}
\date{}

\begin{document}
\maketitle

\begin{abstract}
We simplify algorithms for Knapsack, output-sensitive Subset Sum, and
near-convex min-plus convolution.

For bounded Knapsack, we give a deterministic algorithm using
$O(N+W^2\log^3(W+2))$ arithmetic and comparison operations,
where $W$ is the maximum item weight and $N$ counts input records
with binary-encoded multiplicities. Following Bringmann's approach,
we partition the items and bound the weight added or removed within
each part when correcting a greedy solution to an optimum. These
bounds keep the dynamic-programming tables small. The same analysis
gives the corresponding bound with maximum profit in place of weight.

For multiple-choice Knapsack, we give a randomized
$\widetilde O(N+w^2\min\{r,w\})$ algorithm, where $N$ counts alternatives,
$r$ counts classes, and $w$ is the maximum within-class weight range.
Randomly grouping classes exploits cancellation between positive and
negative weight changes.

For Subset Sum of $n$ nonnegative integer vectors in any fixed dimension
$d$, we obtain expected time $\widetilde O(n+s\sqrt n)$, where $s$
counts attainable sums in the target box. With high probability,
all sums are returned within the same bound. This improves the
$\widetilde O(n+s n^{d/(d+1)})$ bound of Bringmann, Fischer, and Nakos
for $d>1$. The key step computes a sumset inside a box without
generating the potentially much larger unrestricted sumset.

Finally, we simplify the $\widetilde O(N(D+1))$ algorithm for
min-plus convolution of integer arrays of total length $N$, where
$D$ is the sum of their maximum deviations above convex arrays.
The deviations can change the minimizing pairs substantially, but
restrict relevant candidate values to short intervals.
\end{abstract}

\paragraph{AI disclosure.}
The algorithms and proofs in this paper were
developed by OpenAI's GPT-6 Astra Pro through an
iterative process directed by the author.
The author also used the aformentioned model to assist with drafting and revising the exposition. The author assumes responsibility for all content.

\section{Introduction}\label{sec:introduction}

The standard dynamic programs for Knapsack and Subset Sum index their
states by weight, up to the capacity or target. This can be much larger
than the structure of the instance warrants. Small item weights can
make an optimum close to a greedy choice even when the capacity is
large; few attainable sums can make a sparse representation preferable
even when the target is large. In both cases, the difficulty is not
only to identify the relevant states, but to organize the computation
without producing many more intermediate states.

Our main focus is Knapsack with maximum item weight $W$.
Bringmann~\cite{Bringmann23} and Jin~\cite{Jin24} independently gave
algorithms nearly quadratic in $W$, apart from reading the input.
A proximity bound is the starting point: some optimum differs from
a prefix based on density in only $O(W)$ items. But these changes can have
total weight $\Theta(W^2)$. Processing all $W$ possible item weights
over a table of that width would still take cubic time. Bringmann's
approach avoids this by partitioning the items and bounding the changes
in each part separately. We give a simpler construction of this
partition and a direct analysis of the resulting dynamic program.

We also revisit multiple-choice Knapsack, output-sensitive Subset Sum,
and near-convex min-plus convolution. For Knapsack, one-dimensional
Subset Sum, and near-convex convolution, the contribution is a simpler
algorithm or proof of a known bound. For multiple-choice Knapsack and
higher-dimensional Subset Sum, the constructions also improve the
running time. We distinguish these contributions below.

\subsection{Results and Prior Work}\label{sec:results}

\paragraph{Knapsack with small weights or profits.}

In bounded Knapsack, a record $(w_i,p_i,u_i)$ represents $u_i$
identical items of weight $w_i$ and profit $p_i$. A solution chooses
multiplicities to maximize total profit subject to a capacity bound.
Let $N$ be the number of records and $W=\max_i w_i$.

\begin{restatable}[Bounded Knapsack]{theorem}{knapsackmain}
\label{thm:knapsack}
For positive integer weights at most $W$, bounded Knapsack can be
solved deterministically in
\[
    O\bigl(N+W^2\log^3(W+2)\bigr)
\]
arithmetic and comparison operations. Multiplicities may be
given in binary and profits may be rational. The algorithm returns
the selected multiplicity of each record.
\end{restatable}

The same bound holds with the maximum positive integer profit $P$
in place of $W$ (\cref{cor:knapsack-profit}). Bringmann already solves
bounded Knapsack in $\widetilde O(N+W^2)$ time, while Jin gives a
$O(N+W^2\log^4(W+2))$ bound for the 0--1 problem. Our result
simplifies the proximity-partition approach rather than improving its
near-quadratic dependence on $W$. We prove the required subset-sum
intersection statement from Kneser's theorem and popular pairs, use
it to construct the partition, and perform the dynamic-programming
updates by concave convolution via SMAWK. One local theorem for exact
weights near a density prefix yields both the small-weight and
small-profit bounds.

\paragraph{Multiple-choice Knapsack.}

Here a solution chooses exactly one alternative from each of $r$
classes. Let $N$ be the total number of alternatives and $w$ the
largest difference between two weights in one class. With
$k=\min\{r,2w\}$, our algorithm uses
\[
    O\!\left(N\log(N+2)+w^2k\log(k+2)
                         \log\frac{k+2}{\delta}\right)
\]
operations; the input term is $O(N)$ when each class is sorted by
weight (\cref{thm:multiple-choice}). This work bound holds for every
choice of the randomness. Infeasibility is detected without error;
otherwise the returned choice is feasible and is optimal with
probability at least $1-\delta$.

Pawlewicz~\cite{Pawlewicz26} gives an $O(N+w^4)$ algorithm for
sorted input. Zhou~\cite{Zhou26} subsequently obtains a randomized
$\widetilde O(N+w^{3.5})$ bound in the maximum-weight parameterization,
along with bounds for other input-size regimes. We retain Pawlewicz's
reference choice and reduction to a short list of candidate changes.
The remaining computation assigns classes randomly to a binary tree
and uses narrow dynamic-programming tables at its nodes. This gives
cubic dependence on $w$, up to logarithms, and a smaller bound when
$r$ is small. The same algorithm computes one prescribed coefficient
of a multiple min-plus convolution (\cref{cor:multiple-convolution}).

\paragraph{Output-sensitive Subset Sum.}

Let $X$ be a multiset of $n$ nonnegative integer vectors in dimension
$d$, and let $s$ count the attainable sums in the target box
$[0,t]^d$, including zero. For every fixed $d$, our algorithm has
expected running time $\widetilde O(n+s\sqrt n)$. It always returns
attainable sums, and with probability at least $1-\delta$ it returns
the entire set within the same time bound
(\cref{thm:random-subset-sum}).

Bringmann, Fischer, and Nakos~\cite{BFN25} give this bound in one
dimension and $\widetilde O(n+s n^{d/(d+1)})$ in dimension $d$.
We retain their separation of small and large items, color coding,
and use of sumset submultiplicativity. A box-restricted sumset routine
removes the dimension dependence from the exponent of $n$; a balanced
recursion makes the total-work analysis short. A time-budget procedure
handles both the unknown output size and the expected-time convolution
calls.

In one dimension, supplying the recursion with the 
prefix-restricted sumset routine of Bringmann and Nakos~\cite{BN21}
also recovers $\widetilde O(n+s^{4/3})$. The recursion and its
$s^{4/3}$ accounting work in every fixed dimension, but we do not
obtain a comparable high-dimensional box-restricted sumset computation. Chan~\cite{Chan26} has derandomized both
one-dimensional bounds. Our simplification concerns the randomized
algorithms.

\paragraph{Near-convex min-plus convolution.}

Suppose integer arrays of total length $N$ lie above convex arrays
with respective error bounds $\Delta_f,\Delta_g$, and put
$D=\Delta_f+\Delta_g$. Bringmann and Cassis~\cite{BC23} give an
$\widetilde O(N(D+1))$ algorithm. We retain their geometric
decomposition and reduction to sparse convolution, but count the
possible integer values at each output index directly. Two-corner
tests identify the rectangles that can be processed together, and a
scale-by-scale count bounds the total convolution support
(\cref{thm:near-convex}). This is a simpler proof and presentation
of the known bound. Its numerical assumptions and determinism are
those of the chosen sparse-convolution primitive.

The computational models are specified in \cref{sec:preliminaries}.
In particular, the Knapsack bounds count arithmetic and
comparison operations. The Subset Sum bound uses the word-RAM assumptions of
the sparse-convolution primitive. For fixed dimension, $\widetilde O$
suppresses the logarithmic factors detailed in the theorem statements,
including dependence on the failure probability.

\subsection{Overview of the Algorithms}\label{sec:overview}

\paragraph{Knapsack: the cost of a part determines the partition.}

Let $g$ be the longest feasible prefix in decreasing density order,
and choose an optimum $x$ with a tie-break favoring earlier
occurrences. Their weights differ by less than $W$. An equal-weight
exchange between removals from $g$ and additions to $g$ would
improve profit or the tie-break, so no such nonempty exchange exists.
Giving removals negative signs and additions positive signs produces
a zero-sum-free sequence of corrections. The short walk lemma bounds its length by
$2W-1$. Within each weight class, both solutions choose the most
profitable occurrences. It is therefore enough to retain the last
$2W$ selected and first $2W$ unselected occurrences of each weight,
leaving a kernel of $O(W^2)$ items.

The remaining task is to process this kernel without paying for a
width-$\Theta(W^2)$ table at every weight. Suppose a part $I$ has
$d_I$ distinct weights and its changed occurrences have total weight
at most $\Delta_I$. Equal-weight items give a concave profit profile,
so SMAWK performs one weight-class update in time linear in the
retained table width. If at most $\kappa$ parts have change bounds
in any one dyadic scale, processing parts in increasing scales keeps
the accumulated radius within a factor $O(\kappa)$ of the current
scale. The total work is then
\[
    O\!\left(\kappa\sum_I d_I\Delta_I\right).
\]
Thus the partition must give small change bounds to parts containing
many distinct weights.

To construct such parts, consider $s$ remaining occurrences with
weight multiplicities in $(m/2,m]$. They use fewer than $2s/m$
distinct weights. Take the larger side of the fixed boundary of $g$,
assign its half farther from that boundary to a part, and leave the
nearer half as a buffer. Suppose the part lies inside $g$. If a
majority of the buffer is selected by $x$, the ordered-exchange rule
forbids a common positive subset sum between that majority and the
removals in the part. If a majority is unselected, it instead forbids
a common positive subset sum between that majority and the additions
to $g$. The subset-sum intersection theorem bounds removals in the
first case and additions in the second. Their weights differ by less
than $W$, so both cases give $\Delta_I=O(mW^2/s)$ when the buffer
is large enough. Consequently $d_I\Delta_I=O(W^2)$.

The construction removes a constant fraction of the current occurrences
at each step. There are logarithmically many multiplicity levels and
steps per level. Within a level, the bounds $mW^2/s$ increase
geometrically, so only constantly many parts lie in any one DP scale.
Including the small remainders, this gives
$\sum_I d_I\Delta_I=O(W^2\log^2(W+2))$ and
$\kappa=O(\log(W+2))$.

\paragraph{Multiple-choice Knapsack: make cancellation occur before merging.}

The reference choice is optimal at its own weight and leaves less
than $w$ unused capacity. A minimal correction to an optimum changes
only $O(w)$ classes: a weight-zero collection of changes could be
undone without decreasing profit. Keeping the best changes at each
weight difference leaves $O(wk)$ candidates.

Although the correction is short, its changes are signed. A sequential
table may need width $\Theta(kw)$ if it processes the positive
changes before the negative ones that cancel them. Randomly assigning
classes to tree leaves makes the contribution of a fixed optimum
concentrate inside each subtree. Leaf occupancies control the partial
sums during the sequential leaf computations; Bernstein's inequality
controls the signed sums at internal nodes. The resulting widths grow
like the square root of the number of descendant leaves. Pairwise
merging then has essentially the same total cost at every tree level.

\paragraph{Subset Sum: restrict the pairs before computing their sums.}

For $C=(A+B)\cap[0,t]^d$, computing $A+B$ and filtering afterward
can be much more expensive than listing $C$. Instead, partition the
valid input pairs into products $A_i\times B_i$ containing only
pairs whose sums fit in the box.

In one coordinate, sort $B$. For $x\in A$, let $q(x)$ count the
entries that can be added to $x$. An entry $y\in B$ is feasible
exactly when its zero-based rank is less than $q(x)$. The first
binary digit where the two numbers differ certifies the comparison:
the smaller starts with $p0$ and the larger with $p1$ for a common
prefix $p$. Grouping by this prefix gives products of input subsets.
Doing so in each coordinate gives the multidimensional decomposition,
with each point copied into only a polylogarithmic number of products
for fixed $d$.

Each $A_i+B_i$ lies in $C$ and has size at most
$\min\{|A_i||B_i|,|C|\}$. Taking the geometric mean of these two
bounds and applying Cauchy--Schwarz controls the sum of the support
sizes. Sparse convolution therefore computes $C$ in expected time
\[
    \widetilde O\!\left(|A|+|B|+\sqrt{|A||B||C|}\right)
\]
for fixed $d$. In particular, inserting a class of $m$ items costs
$\widetilde O(m+s\sqrt m)$ when both the old and new subset-sum
sets have size at most $s$.

\paragraph{The child target controls both correctness and merge size.}

At a Subset Sum call with target $t$, partition the small items into
$k$ balanced random children and give each target
$\lfloor t/(k-1)\rfloor$. A fixed witness sends an expected $1/k$
share of each coordinate to a child; the slightly larger target
provides slack for concentration. At the same time, any sum omitting
one child fits in the parent box. If $T_i$ is the true child output,
sumset submultiplicativity yields
\[
    |T_1+\cdots+T_k|\le s^{k/(k-1)}=O(s)
\]
for logarithmic $k$. We can therefore compute the full child merge
before truncating it.

A second sumset inequality gives
$\sum_i(|T_i|-1)\le2(s-1)$. Since each child has at most $2n/k$
items, the potential $(s-1)\sqrt n$ contracts by a factor at most
$2\sqrt{2/k}<1$. This bounds the work over the entire recursion.
A feasible witness uses only polylogarithmically many large items;
color coding separates them into classes, and the restricted-sumset
routine inserts one class at a time. A time-budget procedure chooses
$k$ without knowing $s$.

\paragraph{Near-convex convolution: count values, not minimizing pairs.}

A bounded perturbation can move a minimizing pair anywhere along its
output diagonal. The useful restriction is instead on its value.
Writing $H=F\mc G$, discard every pair whose reference value
exceeds $H_{i+j}+D$. Every true minimizer remains, and every
remaining pair has perturbed value in
$[H_{i+j},H_{i+j}+2D]$. Within a relevant rectangle, each output
index therefore has only $O(D+1)$ candidate integer values.

Convexity makes the relevant pairs form a band between two monotone
grid paths. A dyadic decomposition accepts rectangles wholly inside
the band and subdivides only along its boundaries. At each scale,
the accepted rectangles have total side length $O(N)$, so their
graph sumsets have total support $O(N(D+1))$. This gives the bound
after summing over the logarithmically many scales.

\paragraph{Organization.}

\Cref{sec:preliminaries} gives the shared tools, including the
subset-sum intersection proof. The algorithms for Knapsack,
multiple-choice Knapsack, Subset Sum, and near-convex convolution
appear in \cref{sec:knapsack,sec:multiple-choice,sec:subset-sum,sec:near-convex},
respectively.

\section{Preliminaries and Additive Combinatorics}\label{sec:preliminaries}
Write $[u]=\{1,\ldots,u\}$ for a positive integer $u$. Intervals
indexing arrays or integer sets include their integer endpoints.
All logarithms are base two, except for $\ln$.
A multiset consists of distinct \emph{occurrences}, which may have
equal values. A submultiset chooses occurrences without replacement.
For a multiset $X$ of integers or vectors, write
\[
 \Sigma(X)=\sum_{x\in X}x,\qquad
 \SubSums{X}=\{\Sigma(Y):Y\subseteq X\}.
\]
For nonnegative vectors in dimension $d$, put
$\SubSums{X,t}=\SubSums{X}\cap[0,t]^d$; the scalar case has $d=1$.
The empty choice is allowed, so every subset-sum set contains zero.
For finite sets, $A+B=\{a+b:a\in A,b\in B\}$ is their sumset.
In contrast to a subset sum, $hA$ permits repetition: it is the set
of sums of any $h$ elements of $A$. We set $0A=\{0\}$.

For arrays on integer intervals, define ordinary, min-plus, and
max-plus convolution by
\[
 (a*b)_k=\sum_{i+j=k}a_i b_j,\qquad
 (a\mc b)_k=\min_{i+j=k}(a_i+b_j),\qquad
 (a\xc b)_k=\max_{i+j=k}(a_i+b_j).
\]
An infeasible minimum is $+\infty$ and an infeasible maximum is
$-\infty$. A finite array is convex if its adjacent differences
are nondecreasing, and concave if they are nonincreasing.

\paragraph{Computational model.}
For the two Knapsack problems we count arithmetic and comparison
operations. Profits may be rational and multiplicities may be binary
encoded. The Subset Sum bound uses a word RAM,
as specified in its statement. The near-convex result inherits the
numerical assumptions
of its sparse-convolution operation. The notation $\softO$ suppresses
polylogarithmic factors. We specify the parameters of these factors
where they affect the guarantee.

Our randomized algorithms can miss an optimum or an attainable sum,
but every choice or sum they return is valid.
We use Bernstein's inequality: if independent centered variables
$Z_i$ satisfy $|Z_i|\le b$ and
$\sum_i\mathbb E Z_i^2\le V$, then
\begin{equation}\label{eq:bernstein}
 \Pr\!\left[\left|\sum_iZ_i\right|>
       \sqrt{2VL}+\frac{2bL}{3}\right]\le2e^{-L}
 \qquad(L>0).
\end{equation}

\subsection{Convolutions}\label{sec:primitives}
For integer sets $A,B$, the coefficient of $x^z$ in
$(\sum_{a\in A}x^a)(\sum_{b\in B}x^b)$ is positive exactly when
$z\in A+B$. Sparse nonnegative convolution therefore lists $A+B$
by computing the nonzero coefficients of this product. If the output has $s$
nonzero positions, deterministic algorithms take
$s\,\operatorname{polylog} U$ time, where $U$ bounds the coordinate
universe; see \cite[Lemma 2.2]{BFN25}. Under the usual word-size
assumptions, the Las Vegas algorithm of Jin and Xu~\cite{JX24}
takes expected time $O(s\log(s+2))$ with probability $1 - \frac{1}{s}$. After computing a
sumset, we discard coefficients and retain only its support.

For inputs in
$[0,U]^d$, encode a vector in base $2U+1$. Coordinatewise addition
causes no carries, so the scalar and vector sumsets have the same
cardinality. A scalar encoded index occupies $O(d)$ words; simulating
the scalar algorithm incurs a polynomial overhead in $d$.

For Knapsack we need max-plus convolution with a concave
array. SMAWK finds row maxima in a totally monotone $r\times c$
matrix in $O(r+c)$ comparisons, given constant-time entry access
\cite{SMAWK}. The next formulation allows both signed indices and
restriction to an output interval. The feasibility conditions require some care: simply padding a
matrix with $-\infty$ need not preserve total monotonicity.

\begin{lemma}[Concave convolution on an output interval]
\label{lem:concave-convolution}
Let $a$ be indexed by an integer interval $I$, with entries in
$\R\cup\{-\infty\}$. Let $b$ be finite and concave on a nonempty
integer interval $J$, with constant-time entry access. The values of
$a\xc b$ on an integer interval $T$, together with maximizing
indices, can be computed in $O(|I|+|T|)$ operations.
Negation gives the corresponding convex min-plus statement.
\end{lemma}
\begin{proof}
Discard columns $i$ for which $a_i=-\infty$. If no column remains,
every output is $-\infty$. Write $J=[u,v]$. For row $t$ and a
remaining column $i$, let $q=t-i$ and clamp it to
$c=\max\{u,\min\{v,q\}\}$. Use the lexicographically ordered entry
\[
 K_{t,i}=\bigl(-|q-c|,\ a_i+b_c\bigr).
\]
A feasible entry has first coordinate zero, and every infeasible
entry has negative first coordinate. Thus the row maximum first
finds a feasible transition, if one exists, and then maximizes its
profit.

As a function of $q$, the pair $(-|q-c|,b_c)$ has successive
differences $(1,0)$ to the left of $J$, $(0,b_{q+1}-b_q)$ within
$J$, and $(-1,0)$ to its right. They are nonincreasing in
lexicographic order. Thus $K$ is anti-Monge:
\[
 K_{t,i}+K_{t',i'}\ge K_{t,i'}+K_{t',i}
 \qquad(t<t',\ i<i').
\]
Deleting columns preserves this property. Apply SMAWK with a
consistent rule for ties. A row maximum with first coordinate zero
gives a feasible maximizing pair. A negative first coordinate means
that no transition in that row is feasible.
\end{proof}

\subsection{A Short Walk}\label{sec:walk}
Suppose we have some sequence where additions are assigned positive weights and removals negative
weights. If some nonempty collection sums to zero, undoing that
collection preserves the total weight. The exchange arguments below
will rule out this possibility. The next lemma then bounds the number
of changes: order them so that their partial sums stay in an interval
of $2W$ integers, and observe that no partial sum can repeat.

\begin{lemma}[Zero-sum-free walk]\label{lem:walk}
Let $z_1,\ldots,z_s$ be nonzero integers in $[-W,W]$ whose sum has
magnitude less than $W$. If no nonempty submultiset sums to zero,
then $s\le2W-1$.
\end{lemma}
\begin{proof}
Negate all terms if necessary so their total is $\delta\in[0,W)$.
Starting from $u=0$, take a remaining positive term when $u<W$,
and a remaining negative term when $u\ge W$. Each such step stays
in $[0,2W-1]$. If the required sign is unavailable, append the
remaining terms; the partial sum then moves monotonically to
$\delta$, still inside this interval.
All $s+1$ partial sums are distinct, since a repeated value gives a
nonempty zero-sum submultiset. The interval has only $2W$ integers.
\end{proof}

\subsection{A Subset-Sum Intersection Bound}
\label{sec:intersection}

The walk lemma bounds the number of changes. The Knapsack partition
also needs to bound their total weight inside one part. The following
exchange obstruction provides that bound: if one multiset has sufficiently
many occurrences without too much repetition, then any second multiset
of sufficiently large total weight must share a positive subset sum with it.

\begin{theorem}[Subset-sum intersection]
\label{thm:intersection}
Let $X,Y$ be multisets of integers in $[W]$. Put $r=|X|\ge1$,
and suppose each value occurs at most $\mu$ times in $X$. If
\[
    r^2\ge512\mu W\lceil\log(2r)\rceil,
    \qquad
    \Sigma(Y)\ge256\mu W^2/r,
\]
then $\SubSums{X}\cap\SubSums{Y}$ contains a positive integer.
\end{theorem}

We will use the contrapositive: if $X$ satisfies the first inequality
and the multisets have no common positive subset sum, then
$\Sigma(Y)<256\mu W^2/r$.

The hypotheses on $X$ and $Y$ play different roles. We use the size
and multiplicity bound of $X$ to construct a progression in
$\SubSums{X}$. Once that progression is available, only the total
weight of $Y$ matters.

\paragraph{Why a progression suffices.}

Suppose $\SubSums{X}$ contains consecutive positive multiples of some
$d$. We can combine occurrences of $Y$ into small blocks whose weights
are divisible by $d$. Adding these blocks keeps the running sum aligned
with the progression. If the progression is long enough,
the first running sum to reach its start cannot jump past its end.

The following lemma shows the required length and total weight.

\begin{lemma}[Hitting a progression]
\label{lem:intersection_hit_progression}
Let $Y$ be a multiset of integers in $[W]$, and let $a,d$ be positive
integers with $d\mid a$. If $\Sigma(Y)\ge a+dW$, then
$\SubSums{Y}$ intersects
\[
    \{a,a+d,\ldots,a+(W-1)d\}.
\]
\end{lemma}

\begin{proof}
Among any $d$ occurrences of $Y$, consider the $d+1$ prefix sums,
including the empty prefix. Two are equal modulo $d$, so the
intervening nonempty block has weight divisible by $d$ and at most
$dW$.

Repeatedly extract such a block while at least $d$ occurrences remain.
Write the extracted block weights as $q_1,\ldots,q_k$. They are positive
multiples of $d$, each at most $dW$. Fewer than $d$ occurrences remain
unused, so
\[
    \sum_{j=1}^k q_j
    \ge \Sigma(Y)-(d-1)W
    > \Sigma(Y)-dW
    \ge a.
\]
Take the shortest prefix of blocks whose total weight $q$ is at least
$a$. The preceding total is less than $a$, and the last block has
weight at most $dW$, so
\[
    a\le q<a+dW.
\]
Since both $a$ and $q$ are divisible by $d$, this places $q$ among
$a,a+d,\ldots,a+(W-1)d$. The blocks are disjoint, so $q$ is a subset
sum of $Y$.
\end{proof}

Thus our task is to construct $W$ consecutive positive multiples of
$d$ inside $\SubSums{X}$, with first point $a$ such that
$a+dW\le\Sigma(Y)$.

\paragraph{From short sums to subset sums.}

We construct the progression using popular pair sums, as in dense
Subset Sum~\cite{BW21}. First, Kneser's theorem gives a progression
formed from short sums of elements of a set $B$, allowing repetition.
We then choose $B$ so that each of its elements has many disjoint
representations as a pair of occurrences of $X$. These alternatives
let us replace every use of an element of $B$ by an unused pair.

\begin{theorem}[Kneser's theorem; see \cite{DeVos13}]
\label{thm:kneser}
For finite nonempty subsets $A,B$ of an abelian group, let
$H=\{h:A+B+h=A+B\}$. Then
\[
    |A+B|\ge |A+H|+|B+H|-|H|.
\]
\end{theorem}

\begin{lemma}[A progression of short sums]
\label{lem:homogeneous}
Let $B\subseteq[M]$ be nonempty, put $b=|B|$, and let $d=\gcd(B)$.
Then $\bigcup_{1\le h\le\lfloor8M/b\rfloor}hB$ contains $M$
consecutive positive multiples of $d$.
\end{lemma}

\begin{proof}
After dividing out the gcd, work modulo the largest value. We first
reach every residue with few summands. Adding copies of the largest
value then fills an interval of integers.

Put $D=B/d$, $m=\max D$, and $A=D\bmod m$.
The $b$ values of $D$ give distinct residues, including zero, and
$A$ generates $\Z/m\Z$. If a nonempty sumset $S+A$ is proper,
its stabilizer $H$ is proper. Since $0\in A$ and $A$ generates
the group, $A$ meets at least two $H$-cosets. Consequently,
\[
    |A+H|-|H|
    \ge \frac{|A+H|}{2}
    \ge \frac b2,
\]
and Kneser's theorem gives $|S+A|\ge|S|+b/2$.
Starting from $\{0\}$, it follows that
$kA=\Z/m\Z$ for $k=\lceil2m/b\rceil$; the case $m=1$ is
immediate.

Choose one $x\in kD$ in each residue class. Every such $x$ is at
most $km$. For $z\in[km,km+M-1]$, choose the representative of
its residue and write $z=x+qm$ with $q\ge0$. Because $m\in D$,
this is a representation using at most
\[
    k+q
    \le2k+\lceil M/m\rceil
    \le6m/b+M/m+1
    \le8M/b
\]
terms of $D$. The last inequality uses $b\le m\le M$.
Multiplying by $d$ gives the claimed progression.
\end{proof}

\begin{proof}[Proof of \cref{thm:intersection}]
Regard equal-valued occurrences of $X$ as distinct objects, and let
$m_a$ be the multiplicity of value $a$. For each sum $z$, let $t_z$
be the maximum number of pairwise disjoint occurrence pairs with sum
$z$. Thus $t_z\le\lfloor r/2\rfloor$.

Disjointness here is only required among pairs with the same sum.
Families corresponding to different sums may overlap. We will account
for these overlaps when we turn a short sum into a subset sum.

\paragraph{There are many disjoint pair representations in total.}

For distinct values $a,b$, we can form $\min\{m_a,m_b\}$ disjoint
pairs with sum $a+b$. Since $m_a,m_b\le\mu$,
\[
    \min\{m_a,m_b\}\ge \frac{m_am_b}{\mu}.
\]
For pairs using the same value $a$, we can form
$\lfloor m_a/2\rfloor$ disjoint pairs, and
\[
    \left\lfloor\frac{m_a}{2}\right\rfloor
    \ge \frac{m_a(m_a-1)}{2\mu}.
\]

For a fixed sum $z$, distinct complementary value classes do not
overlap: a value $a$ can only be paired with $z-a$. We may therefore
take all the pairs just described simultaneously for that sum.
Summing over $z$ gives
\begin{equation}\label{eq:intersection_pair_mass}
\begin{aligned}
    \sum_z t_z
    &\ge
    \frac1\mu\sum_{a<b}m_am_b
    +\frac1{2\mu}\sum_a m_a(m_a-1)\\
    &=\frac{r(r-1)}{2\mu}
    \ge\frac{r^2}{4\mu}.
\end{aligned}
\end{equation}
The last inequality uses $r\ge2$, which follows from the density
hypothesis.

The interpretation is that the $\binom r2$ occurrence pairs yield
at least $\binom r2/\mu$ disjoint representations when counted
separately for each sum. Bounded multiplicity prevents too many
representations of one sum from depending on the same occurrences.

\paragraph{Choosing a threshold with enough values and enough alternatives.}

For an integer threshold $v\ge1$, define
\[
    B_v=\{z:t_z\ge v\},
    \qquad b_v=|B_v|.
\]
Every element of $B_v$ has at least $v$ disjoint pair representations.
On the other hand, \cref{lem:homogeneous}, applied with $M=2W$,
uses at most $\lfloor16W/b_v\rfloor$ elements of $B_v$ to represent
each progression point.

This exposes the two requirements on our threshold. We need $b_v$
large so that the representations are short and the progression starts
early. We also need $vb_v$ large so that the number of pairs used in
one representation is much smaller than the number $v$ of available
alternatives for each pair sum.

Put $L_r=\lceil\log(2r)\rceil$. We claim that some
$1\le v\le\lfloor r/2\rfloor$ satisfies
\begin{equation}\label{eq:popular-pairs}
    b:=b_v
    \ge
    \frac{r}{4\mu}
    +\frac{r^2}{8\mu L_rv}.
\end{equation}
The two terms on the right guarantee the two requirements separately:
the first bounds $b$ from below, while the second bounds $vb$ from
below.

To prove the claim, count each sum $z$ once for every threshold
$v\le t_z$. This gives the identity
\[
    \sum_{v=1}^{\lfloor r/2\rfloor}b_v=\sum_z t_z.
\]
If \eqref{eq:popular-pairs} failed at every threshold, then
\[
\begin{aligned}
    \sum_z t_z
    &<
    \frac{r}{4\mu}\left\lfloor\frac r2\right\rfloor
    +\frac{r^2}{8\mu L_r}
       \sum_{v=1}^{\lfloor r/2\rfloor}\frac1v\\
    &\le
    \frac{r^2}{8\mu}+\frac{r^2}{8\mu}
    =\frac{r^2}{4\mu},
\end{aligned}
\]
contradicting \eqref{eq:intersection_pair_mass}. Here we used
$\sum_{v\le r/2}1/v\le L_r$.

Fix such a threshold $v$. By \eqref{eq:popular-pairs} and the density
hypothesis,
\[
    b\ge\frac{r}{4\mu},
    \qquad
    vb\ge\frac{r^2}{8\mu L_r}\ge64W.
\]

\paragraph{Realizing the progression with distinct occurrences of $X$.}

Apply \cref{lem:homogeneous} to $B_v\subseteq[2W]$.
It gives a progression $P$ of $2W$ consecutive positive multiples
of $d=\gcd(B_v)$, each represented by at most
$H=\lfloor16W/b\rfloor$ elements of $B_v$. The bounds above give
\[
    H\le\frac{64\mu W}{r},
    \qquad
    H\le\frac v4.
\]

We must still check that these representations are subset sums of
$X$: the short-sum lemma allows repetitions in $B_v$, whereas an
occurrence of $X$ may be used only once.

Fix one representation $z_1+\cdots+z_h$ of a point of $P$, where
$h\le H$. For each $z_i$, fix a family of $v$ disjoint occurrence
pairs having sum $z_i$. Choose one pair from each family in order.

Before choosing the $i$th pair, exactly $2(i-1)$ occurrences have
been used. Because the $v$ candidates for $z_i$ are pairwise
disjoint, each used occurrence rules out at most one candidate.
Thus at most $2(i-1)$ candidates are unavailable. But
\[
    2(i-1)<2H\le v/2<v,
\]
so an unused pair remains.

This works even when some of the $z_i$ are equal. Hence every point
of $P$ has a representation using distinct occurrences of $X$, and
\[
    P\subseteq\SubSums{X}.
\]

\paragraph{The progression is early enough and its step is small enough.}

Write $a=\min P$. Every point of $P$ is a sum of at most $H$ pair
sums, each at most $2W$, so $\max P\le2WH$.

Also, $B_v$ contains $b$ distinct positive multiples of $d$, all at
most $2W$. Therefore $bd\le2W$. Together these bounds give
\begin{equation}\label{eq:intersection-progression}
    a\le\max P\le2WH\le\frac{128\mu W^2}{r},
    \qquad
    d\le\frac{2W}{b}\le\frac{8\mu W}{r}.
\end{equation}
In particular,
\[
    a+dW
    \le\frac{136\mu W^2}{r}
    \le\frac{256\mu W^2}{r}
    \le\Sigma(Y).
\]

Apply \cref{lem:intersection_hit_progression} to the first $W$ points
of $P$. It gives a positive subset sum of $Y$ lying in $P$, and
every point of $P$ is a subset sum of $X$. This proves the theorem.
\end{proof}

\section{Knapsack with Small Weights or Profits}\label{sec:knapsack}
The input consists of $N$ records $(w_i,p_i,u_i)$, each representing
$u_i$ identical items of weight $w_i$ and profit $p_i$. We seek
integers $0\le x_i\le u_i$ maximizing $\sum_i p_i x_i$ subject to
$\sum_i w_i x_i\le C$, where $C\ge0$. Throughout this section, weights are positive
integers and multiplicities may be binary encoded. For a choice of
occurrences $x$, write $\wt(x)$ and $\prf(x)$ for its weight and
profit.

Order items by nonincreasing density $p_i/w_i$ and let $g$ be
the longest prefix of weight at most $C$. After discarding
nonpositive profits, both $g$ and an optimum leave less than $W$
unused capacity, unless all items fit. Thus the optimum has weight
$\wt(g)+\delta$ for some integer $|\delta|<W$. We will compute
the best profit for every such $\delta$ in one dynamic program.

The construction uses only the fact that $g$ is a density prefix.
We therefore state it for any such prefix and allow profits of
either sign. This form will give the small-profit bound as well:
use profit as the new weight and negative weight as the new profit.

\begin{lemma}[Optima near a density prefix]\label{lem:local-profile}
Suppose $w_i\in[W]$ and profits are arbitrary rationals. Given a
prefix $g$ in nonincreasing density order, one can compute
\[
 H(z)=\max\{\prf(x):\wt(x)=z\}
 \qquad\text{for every integer }|z-\wt(g)|<W
\]
in $O(N+W^2\log^3(W+2))$ operations. Unattainable weights receive
value $-\infty$. A witness for any one finite entry can be recovered
within the same bound.
\end{lemma}

The proof first confines all changes from $g$ to $O(W^2)$
occurrences. On this kernel, we describe how to process a partition
with a known change bound for each part. This identifies the cost
that the partition must control. We then construct the partition
and substitute its bounds. We use the same kernel and partition
for every target $z$ in the lemma.

\subsection{Ordered Exchanges and a Kernel Near the Prefix}
\label{sec:knapsack-kernel}
Fix the density order, breaking ties by record and occurrence index.
For each attainable target $z$ in the lemma, choose an optimum $x_z$
whose incidence vector is lexicographically greatest in this order.
It satisfies the following exchange rule. If nonempty sets
$A\subseteq\overline{x_z}$ and $B\subseteq x_z$ have every
occurrence of $A$ before every occurrence of $B$, then
\begin{equation}\label{eq:ordered-exchange}
 \wt(A)\ne\wt(B).
\end{equation}
If their weights were equal, replace $B$ by $A$. Every item of
$A$ has density at least that of every item of $B$, so the
exchange cannot decrease profit. If profit stays the same, the
first changed coordinate is an earlier item that becomes selected,
improving the tie-break. Both conclusions contradict the choice of
$x_z$. The argument allows negative profits. Applying it to two
items of the same weight also shows that $x_z$ selects a prefix
within each weight class.

Let $E^-=g\setminus x_z$ be the removals from $g$ and
$E^+=x_z\setminus g$ the additions. All removals precede all
additions, and
\begin{equation}\label{eq:knapsack-balance}
 \wt(E^+)-\wt(E^-)=z-\wt(g)\in(-W,W).
\end{equation}
Assign an addition its positive weight and a removal its negative
weight. A nonempty zero-sum submultiset would supply equal-weight
subsets of $E^-$ and $E^+$. The first is unselected and lies
entirely before the selected second, contradicting
\eqref{eq:ordered-exchange}. Applying \cref{lem:walk} to these
signed changes gives
\begin{equation}\label{eq:knapsack-proximity}
 |g\triangle x_z|\le2W-1.
\end{equation}

The location of these changes matters as much as their number.
Within each weight class, both $g$ and $x_z$ choose a prefix.
Their symmetric difference is therefore adjacent to the end of the
prefix chosen by $g$. Retain, for each weight,
the last $2W$ occurrences selected by $g$ and the first $2W$
not selected by $g$, or all available occurrences if fewer exist.
Fix every other choice as in $g$. Equation~\eqref{eq:knapsack-proximity}
shows that this common kernel preserves every $x_z$, and its size
is at most $4W^2$.

\paragraph{Working with binary multiplicities.}
We expand only the kernel. When $g$ must first be found for a
capacity constraint, use linear-time selection on record
densities~\cite{BFPRT73}. At a median density, compute
the total weight of all occurrences on the earlier side. Comparing
this mass with the remaining capacity determines which half of the
records contains the break. Each step discards at least half the
records; division determines the chosen multiplicity in the final
record. Ties use the fixed record order.

Within each weight class, the same selection procedure, now using
occurrence counts as masses, finds the required selected suffix
and unselected prefix. Across all classes this takes $O(N+W)$
operations. Expand and sort only the retained occurrences, storing
their original record identifiers. This costs
$O(W^2\log(W+2))$.

Subtract the weight and profit of all fixed selected occurrences.
The restricted $g$ remains a density prefix and
\eqref{eq:ordered-exchange}--\eqref{eq:knapsack-balance} remain valid.
For the rest of the proof let $n\le4W^2$ be the kernel size and
$L=\lceil\log(2n)\rceil$. An empty kernel is immediate.

\subsection{The Dynamic Program Needs Local Change Bounds}
\label{sec:knapsack-dp}
The global bound \eqref{eq:knapsack-proximity} permits changes of
total weight $\Theta(W^2)$. Applying one update for each possible
weight over that whole range would still cost $\Theta(W^3)$.
We instead partition the kernel into parts $I$ with integer bounds
$\Delta_I\ge1$ such that
\begin{equation}\label{eq:part-proximity}
 \wt\bigl((g\triangle x_z)\cap I\bigr)\le\Delta_I
 \qquad\text{for every attainable target }z.
\end{equation}
The symmetric difference counts both additions and removals with
positive weight. For example, additions of total weight $a$ and
removals of total weight $b$ contribute $a+b$ to this bound,
even if their net change $a-b$ is small. Hence any partial
correction from $I$, regardless of processing order, has signed
weight change of magnitude at most $\Delta_I$.

Let $d_I$ count the distinct weights in $I$. A DP confined to
weight changes in $[-\Delta_I,\Delta_I]$ would take
$O(d_I\Delta_I)$ time: one linear-time concave-convolution update
for each weight. When combining parts, the table must also contain
the changes made by previously processed parts.

Set $h(I)=\lceil\log\Delta_I\rceil$ and process these scales in
increasing order. If at most $\kappa$ parts have any one scale,
then all parts processed through scale $h$ have total change bound
less than $2\kappa2^h$. Thus an update for a part of bound
$\Delta_I$ needs work only $O(\kappa\Delta_I)$.

\begin{lemma}[Dynamic programming from local bounds]
\label{lem:partition-dp}
Suppose \eqref{eq:part-proximity} holds and at most $\kappa$ parts
have any one scale. The operation count for computing all values in
\cref{lem:local-profile} on the kernel, including a witness for
one requested value, is
\[
 O\!\left(n\log(n+2)+\kappa\sum_I d_I\Delta_I\right).
\]
\end{lemma}
\begin{proof}
Process scales in increasing order. For any $x_z$, the total weight
of its changes in parts of scale at most $h$ is bounded by
$\kappa\sum_{j=0}^h2^j<2\kappa2^h$.
We can therefore retain only signed weight changes in
$[-R_h,R_h]$, where $R_h=2\kappa2^h$. This interval contains
every partial correction of $x_z$ while processing scale $h$.

For each weight $w$ and scale $h$, put all corresponding
occurrences into one group $G$, ordered by decreasing profit.
Both $g$ and $x_z$ choose prefixes of this order. If $g$ chooses
$r_G=|g\cap G|$ occurrences and $P_G(j)$ is the sum of the
$j$ largest profits, then changing the selected count by $q$
changes profit by
\[
 B_G(q)=P_G(r_G+q)-P_G(r_G),
 \qquad -r_G\le q\le|G|-r_G.
\]
Increasing $q$ by one adds the next profit in the sorted group.
These increments are nonincreasing, so $B_G$ is concave, including
at negative $q$. Sorting the groups and forming the prefix sums
costs $O(n\log(n+2))$ in total.

Initialize $F(0)=0$, with all other states infeasible. At each stage, $F(t)$ is the best retained profit change at signed
weight change $t$. Only processed groups may differ from $g$;
all other groups still use their choices from $g$.
For a group of weight $w$ at scale $h$, the update is
\begin{equation}\label{eq:knapsack-dp}
 F_{\rm new}(t)=\max_q
       \{F_{\rm old}(t-wq)+B_G(q)\},
 \qquad |t|\le R_h.
\end{equation}
A transition changes $t$ by a multiple of $w$, so different
residue classes modulo $w$ can be processed separately. Writing
$t=r+wj$ in one residue class turns \eqref{eq:knapsack-dp} into
$F_{\rm new}(r+wj)=\max_q\{F_{\rm old}(r+w(j-q))+B_G(q)\}$.
This is max-plus convolution with the concave array $B_G$.
By \cref{lem:concave-convolution}, its cost is linear in the
number of input and output states in the residue class. Summed
over the residues that occur, this is $O(R_h)$, even if
$w>R_h$.

Every finite state describes a valid choice. Conversely, fix
$x_z$. Its choice in each group occurs in $B_G$, and the absolute
value of any partial signed change is at most the total changed
weight in the parts processed so far. By \eqref{eq:part-proximity},
this lies in the retained interval. Induction over the updates
therefore preserves a value at least as large as the profit of
the corresponding partial choice of $x_z$.
The final value $\prf(g)+F(z-\wt(g))$ is consequently exact:
the desired optimum survives, and every competing finite state is
a feasible choice of the same weight.

If $D_h$ is the number of groups at scale $h$, then
$D_h\le\sum_{I:h(I)=h}d_I$ and $2^{h(I)}<2\Delta_I$.
The update cost is
\[
 \sum_h O(D_hR_h)
 =O\!\left(\kappa\sum_hD_h2^h\right)
 =O\!\left(\kappa\sum_I d_I\Delta_I\right).
\]
Store the maximizing transitions returned by SMAWK. Backtracking
recovers one optimum, and its chosen occurrences are accumulated
by original record. Restoring the choices completes the
witness within the same bound.
\end{proof}

We will construct a partition with
$\kappa=O(L)$ and $\sum_I d_I\Delta_I=O(W^2L^2)$.
Thus a part with many distinct weights must have a correspondingly
small change bound.

\subsection{A Buffer Bounds the Changes in One Part}
\label{sec:knapsack-partition}

We now construct the partition needed by \cref{lem:partition-dp}.
For each part $I$, we want a bound $\Delta_I$ on the total weight
of occurrences in $I$ whose selection differs between $g$ and $x_z$:
\[
    \wt\bigl(I\cap(g\triangle x_z)\bigr)\le\Delta_I.
\]
The partition and its bounds must work simultaneously for all the
solutions $x_z$.

Recall that the occurrences are in a fixed density order and $g$ is
a prefix of that order. The \emph{boundary of $g$} is the position
between the last occurrence in $g$ and the first occurrence outside
$g$.

We maintain a working set $U$ of occurrences not yet assigned to
parts. Deleting an occurrence from $U$ only means assigning it to a
part: it does not change $g$ or $x_z$. Below, ``selected'' and ``unselected'' always refer to
membership in $x_z$.

\paragraph{The target bound at one multiplicity level.}

Suppose $J\subseteq U$ consists of all remaining occurrences of
weights whose current multiplicities lie in $(m/2,m]$, and put
$s=|J|$. Since each such weight has more than $m/2$ occurrences,
$J$ contains fewer than $2s/m$ distinct weights. Thus any part
$I\subseteq J$ has $d_I<2s/m$.

This suggests the change bound we should aim for:
\[
    \Delta_I=O(mW^2/s)
    \qquad\Longrightarrow\qquad
    d_I\Delta_I=O(W^2).
\]
A large working set should therefore give a part with a small
change bound. The buffer is what makes this possible.

\paragraph{Choosing a part and its buffer.}

Take the larger of $J\cap g$ and $J\setminus g$, and call it $S$.
We assign the half of $S$ farther from the boundary of $g$ to the
new part $I$, leaving the other half as a buffer $B=S\setminus I$.
More explicitly:
\begin{itemize}
    \item If $S\subseteq g$, let $I$ be the first
    $\lceil|S|/2\rceil$ occurrences of $S$ in density order.
    \item If $S\subseteq K\setminus g$, let $I$ be the last
    $\lceil|S|/2\rceil$ occurrences of $S$ in density order.
\end{itemize}
Only $I$ is assigned at this step; $B$ remains in $U$.
Writing $q=|I|$, we have $q\ge s/4$ and $|B|\ge q-1$.

\paragraph{Why the buffer controls changes.}

Fix a solution $x_z$, and write
$E^-=g\setminus x_z$ for its removals from $g$ and
$E^+=x_z\setminus g$ for its additions to $g$.
The ordered-exchange property \eqref{eq:ordered-exchange} forbids
a common positive subset sum between an earlier unselected set and
a later selected set: such a sum would allow an equal-weight
exchange in favor of the earlier occurrences.

Consider a part $I\subseteq g$. The relevant sets occur in the order
\[
    I \prec B \prec E^+,
\]
where $\prec$ means that every occurrence of the first set precedes
every occurrence of the second.

At least half of $B$ is either selected or unselected.
If a large portion is selected, it cannot share a positive subset
sum with the earlier unselected occurrences $I\setminus x_z$.
The intersection theorem then bounds the missing weight
in $I$.

If instead a large portion of $B$ is unselected, it cannot share a
positive subset sum with the later selected additions $E^+$.
This bounds the total added weight. The balance estimate
\eqref{eq:knapsack-balance} then bounds the total removed weight,
and hence the missing weight in $I$, at an additional cost of at
most $W$.

Thus either selection pattern in the buffer gives a useful bound.
A part outside $g$ uses the reversed arrangement
$E^-\prec B\prec I$; the proof below spells out both cases.
See \cref{fig:knapsack-buffer}.

\paragraph{When the working set is small.}

The intersection theorem requires enough buffer occurrences:
$s^2$ must be sufficiently large compared with $mWL$, where
$L=\lceil\log_2(2n)\rceil$. If it is not, we assign all of $J$
to one remainder part and use the trivial bound $\Delta_I=Ws$.
Although this bound is weaker, its cost is still small:
\[
    d_I\Delta_I
    <\frac{2s}{m}\cdot Ws
    =\frac{2Ws^2}{m}
    =O(W^2L).
\]

The algorithm uses $c_0=2^{18}$ to separate these two cases.
For a working set $U$, let $m_U(w)$ denote the number of its
occurrences of weight $w$. If a weight's multiplicity falls to
at most $m/2$, all its remaining occurrences wait for a later
multiplicity level.

\begin{algorithm}[H]
\caption{Partitioning the Knapsack kernel}
\label{alg:knapsack-partition}
\begin{algorithmic}[1]
\Require A nonempty kernel $K$ of $n$ occurrences in density order,
reference prefix $g$, and maximum weight $W$.
\Ensure A partition $\mathcal P$ with bounds $\Delta_I$ satisfying
\eqref{eq:part-proximity}.
\State $U\gets K$; $\mathcal P\gets\varnothing$
\State $L\gets\lceil\log_2(2n)\rceil$; $c_0\gets2^{18}$
\For{$j=\lceil\log_2 n\rceil,\ldots,0$}
    \State $m\gets2^j$
    \State $J\gets\{i\in U:m/2<m_U(w_i)\le m\}$
    \While{$J\ne\varnothing$}
        \State $s\gets|J|$
        \If{$s^2\le c_0mWL$}
            \State $I\gets J$; $\Delta_I\gets Ws$
            \Comment{Assign the small remainder.}
        \Else
            \If{$|J\cap g|\ge|J\setminus g|$}
                \State $S\gets J\cap g$
                \State $I\gets$ the first $\lceil|S|/2\rceil$
                occurrences of $S$ in density order
            \Else
                \State $S\gets J\setminus g$
                \State $I\gets$ the last $\lceil|S|/2\rceil$
                occurrences of $S$ in density order
            \EndIf
            \State $B\gets S\setminus I$
            \Comment{The buffer stays in $U$.}
            \State $\Delta_I\gets\lceil c_0mW^2/s\rceil$
        \EndIf
        \State Add $(I,\Delta_I)$ to $\mathcal P$
        \State $U\gets U\setminus I$
        \State Recompute $J$ using the updated multiplicities $m_U(w)$
    \EndWhile
\EndFor
\State \Return $\mathcal P$
\end{algorithmic}
\end{algorithm}

\begin{figure}[t]
\centering
\begin{tikzpicture}[x=1cm,y=1cm,font=\small]
    \draw[->,black!55] (0,3.45) -- (10.6,3.45)
        node[midway,above] {earlier to later in the fixed density order};
    \draw[dashed,black!55] (7.15,-0.45) -- (7.15,2.85);
    \node[above] at (7.15,2.85) {boundary of $g$};

    \draw[fill=black!4] (0,1.5) rectangle (2.9,2.25);
    \draw[fill=blue!10] (3.7,1.5) rectangle (6.6,2.25);
    \draw[fill=black!4] (7.7,1.5) rectangle (10.6,2.25);
    \node[align=center] at (1.45,1.875)
        {$I\setminus x_z$\\unselected};
    \node[align=center] at (5.15,1.875)
        {$B\cap x_z$\\selected};
    \node[align=center] at (9.15,1.875)
        {$E^+$\\selected};
    \draw[->,blue!65!black,thick]
        (5.15,2.30) -- (5.15,2.65)
        -- (1.45,2.65) -- (1.45,2.30);
    \node[anchor=east] at (-0.2,1.875) {(a)};

    \draw[fill=black!4] (0,0) rectangle (2.9,0.75);
    \draw[fill=blue!10] (3.7,0) rectangle (6.6,0.75);
    \draw[fill=black!4] (7.7,0) rectangle (10.6,0.75);
    \node[align=center] at (1.45,0.375)
        {$I\setminus x_z$\\unselected};
    \node[align=center] at (5.15,0.375)
        {$B\setminus x_z$\\unselected};
    \node[align=center] at (9.15,0.375)
        {$E^+$\\selected};
    \draw[->,blue!65!black,thick]
        (9.15,0.80) -- (9.15,1.15)
        -- (5.15,1.15) -- (5.15,0.80);
    \node[anchor=east] at (-0.2,0.375) {(b)};

    \node[below] at (1.45,-0.07) {part $I\subseteq g$};
    \node[below] at (5.15,-0.07) {buffer $B\subseteq g$};
    \node[below] at (9.15,-0.07) {outside $g$};
\end{tikzpicture}
\caption{The two buffer cases for a part $I\subseteq g$.
Selection refers to $x_z$. In each row, the blue set contains
at least half the buffer. Arrows indicate the equal-weight
exchange that would replace later selected occurrences by earlier
unselected ones. Such an exchange is forbidden.
In (a), intersection bounds $\wt(I\setminus x_z)$ directly.
In (b), it bounds $\wt(E^+)$, and the balance inequality then
bounds $\wt(I\setminus x_z)$.}
\label{fig:knapsack-buffer}
\end{figure}
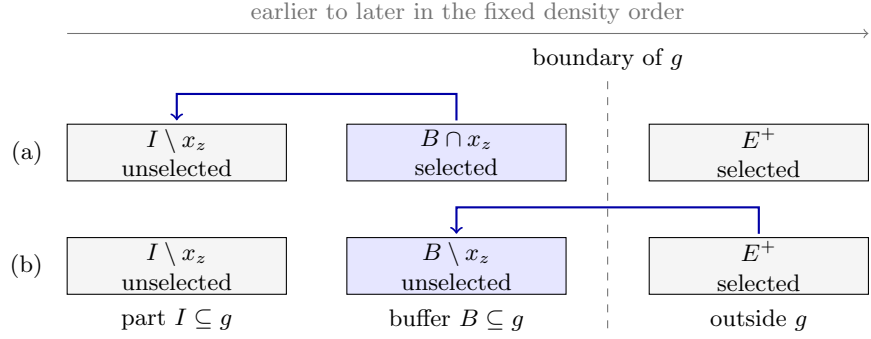

\begin{lemma}[Buffer bound]
\label{lem:knapsack-buffer}
Every part produced by \cref{alg:knapsack-partition} satisfies
\eqref{eq:part-proximity}.
\end{lemma}

\begin{proof}
Fix a profile solution $x_z$. A remainder part $I=J$ has at most
$Ws$ total weight, so its change bound is immediate.

Now consider a buffered part. We first check that either majority
of the buffer is large enough for the intersection theorem, then
identify the set against which to apply it.

\paragraph{The buffer supplies a dense multiset.}

Let $q=|I|$. The construction gives $q\ge s/4$ and $|B|\ge q-1$.
The condition $s^2>c_0mWL$ ensures $q\ge2$. Thus the larger of
$B\cap x_z$ and $B\setminus x_z$, which we denote by $X$, satisfies
\[
    |X|\ge\frac{|B|}{2}
    \ge\frac{q-1}{2}
    \ge\frac q4
    \ge\frac{s}{16}.
\]
Since $X\subseteq J$, each weight occurs at most $m$ times in $X$.
Moreover, $|X|\le n$, so
\[
    |X|^2
    \ge\frac{s^2}{256}
    >1024mWL
    \ge512mW\lceil\log(2|X|)\rceil.
\]
The density hypothesis of \cref{thm:intersection} therefore holds.

Consequently, whenever a multiset $Y$ has no common positive
subset sum with $X$, the contrapositive of that theorem gives
\begin{equation}\label{eq:knapsack-buffer-threshold}
    \wt(Y)
    <\frac{256mW^2}{|X|}
    \le\frac{4096mW^2}{s}.
\end{equation}
Write $\tau=4096mW^2/s$ for this upper bound.

\paragraph{A part inside $g$.}

Suppose $I\subseteq g$. Its changed occurrences are precisely
$I\setminus x_z$, and the order is $I\prec B\prec E^+$.

If $X=B\cap x_z$, take $Y=I\setminus x_z$.
Here $Y$ is unselected and lies entirely before the selected set
$X$. A common positive subset sum would therefore contradict
\eqref{eq:ordered-exchange}. By
\eqref{eq:knapsack-buffer-threshold},
$\wt(I\setminus x_z)<\tau$.

If $X=B\setminus x_z$, take $Y=E^+$.
Now $X$ is unselected and lies entirely before the selected set
$Y$, so the same exchange property gives $\wt(E^+)<\tau$.
Using \eqref{eq:knapsack-balance},
\[
    \wt(I\setminus x_z)
    \le\wt(E^-)
    \le\wt(E^+)+W
    <\tau+W.
\]

Thus in either case the changed weight in $I$ is less than
$\tau+W$.

\paragraph{A part outside $g$.}

Suppose $I\subseteq K\setminus g$. Its changed occurrences are
$I\cap x_z$, and the order is $E^-\prec B\prec I$.

If $X=B\setminus x_z$, take $Y=I\cap x_z$.
The unselected set $X$ lies entirely before the selected set
$Y$, so \eqref{eq:ordered-exchange} and
\eqref{eq:knapsack-buffer-threshold} directly give
$\wt(I\cap x_z)<\tau$.

If $X=B\cap x_z$, take $Y=E^-$.
The unselected removals $Y$ lie entirely before the selected
buffer occurrences $X$. Thus $\wt(E^-)<\tau$, and balance gives
\[
    \wt(I\cap x_z)
    \le\wt(E^+)
    \le\wt(E^-)+W
    <\tau+W.
\]

\paragraph{Absorbing the balance error.}

There are at most $W$ weight values in $J$, each with multiplicity
at most $m$, so $s\le mW$. Hence $W\le mW^2/s$, and therefore
\[
    \tau+W
    \le\frac{4097mW^2}{s}
    \le\left\lceil\frac{c_0mW^2}{s}\right\rceil
    =\Delta_I.
\]
This proves the required bound for the fixed $x_z$. Since the
construction of $I$ and $B$ did not depend on $x_z$, it holds
simultaneously for every $x_z$.
\end{proof}

\subsection{Accounting for All Parts}
\label{sec:knapsack-charging}

The accounting uses two different scales. The \emph{multiplicity
level} $m$ determines which weights are currently processed.
The \emph{DP scale} $h(I)$ groups parts according to the size of
their change bounds $\Delta_I$.

The same geometric decrease controls both. At the start of level
$m$, every remaining weight has multiplicity at most $m$.
During the level, multiplicities only decrease. A weight outside
$J$ therefore cannot enter $J$: its multiplicity is already at
most $m/2$.

Each buffered step assigns at least $s/4$ occurrences to its part.
Recomputing $J$ may remove still more occurrences from the current
level when their weights fall to multiplicity at most $m/2$.
Thus, if $s'$ is the next working-set size at this level,
\[
    s'\le\frac34s.
\]
A remainder step assigns all of $J$ and ends the level.
At the end of level $m$, every remaining multiplicity is at most
$m/2$, establishing the invariant for the next level. In
particular, after level $1$, every occurrence has been assigned.

\begin{lemma}[Partition cost]
\label{lem:knapsack-summation}
The partition satisfies
\[
    \sum_I d_I\Delta_I=O(W^2L^2),
    \qquad
    |\{I:h(I)=h\}|\le4L
    \quad\text{for every }h.
\]
\end{lemma}

\begin{proof}
There are at most $L$ multiplicity levels. Within one level,
the successive buffered sizes decrease by a factor of at least
$4/3$, starting from at most $n$. Hence there are $O(L)$ buffered
parts and at most one remainder part.

\paragraph{Total partition cost.}

At a step with $s=|J|$, the part has $d_I<2s/m$ distinct weights.
For a buffered part,
\[
\begin{aligned}
    d_I\Delta_I
    &<\frac{2s}{m}
      \left(\frac{c_0mW^2}{s}+1\right)\\
    &=2c_0W^2+\frac{2s}{m}
    =O(W^2),
\end{aligned}
\]
using $s\le mW$ to handle the rounding term.

For a remainder part, $\Delta_I=Ws$ and
$s^2\le c_0mWL$, so
\[
    d_I\Delta_I
    <\frac{2Ws^2}{m}
    \le2c_0W^2L.
\]
Thus the $O(L)$ buffered parts and the single remainder part
contribute $O(W^2L)$ per multiplicity level. Summing over the
at most $L$ levels gives $O(W^2L^2)$.

\paragraph{Parts at one DP scale.}

Fix a multiplicity level, and list its buffered parts in order.
If $s_j$ is the working-set size producing the $j$th part, let
$u_j=c_0mW^2/s_j$ be its unrounded change bound.
Since $s_{j+1}\le3s_j/4$, we have
$u_{j+1}\ge4u_j/3$.

The rounded bounds therefore cannot stay in one dyadic interval
for many steps. More precisely, $s_j\le mW$ gives
$u_j\ge c_0W\ge6$, and hence
\[
    u_{j+3}
    \ge\left(\frac43\right)^3u_j
    >2(u_j+1)
    \ge2\lceil u_j\rceil.
\]
Thus the fourth of any four buffered parts has change bound
more than twice that of the first. At most three buffered parts
can lie in a single DP scale $(2^{h-1},2^h]$.

There is at most one additional remainder part at this level.
Each multiplicity level therefore contributes at most four parts
to any fixed DP scale. Summing over the levels gives the bound
$4L$.
\end{proof}

\begin{proof}[Proof of \cref{lem:local-profile}]
Construct the kernel and run \cref{alg:knapsack-partition}.
By the preceding accounting, there are $O(L^2)$ part assignments.

Index the distinct kernel weights once, and retain the occurrences
in density order. At each assignment, a scan of the kernel counts
the current multiplicities and identifies $J$. Further scans select
the appropriate first or last half in density order. Thus the
partition can be constructed in $O(nL^2)$ operations.

The kernel has $n=O(W^2)$ occurrences, so this is
$O(W^2L^2)$ work. By \cref{lem:knapsack-summation}, the partition
has total cost $O(W^2L^2)$ and at most $\kappa=4L$ parts at
any DP scale. Applying \cref{lem:partition-dp} therefore computes
the profile, including reconstruction, in $O(W^2L^3)$ operations.

Finally, $L=O(\log(W+2))$. Adding the input scan proves the lemma.
\end{proof}

\subsection{Returning to the Capacity Problem}
\label{sec:knapsack-consequences}

The local profile solves exact-weight problems near a density
prefix. To solve the capacity problem, it remains to show that
an optimum has a weight in this range.

\begin{proof}[Proof of \cref{thm:knapsack}]
Discard records of zero multiplicity or nonpositive profit.
If all remaining occurrences fit, take them all.

Otherwise, let $g$ be the longest feasible density prefix.
The next occurrence does not fit and has weight at most $W$,
so $g$ leaves slack less than $W$.

An optimum $x^*$ also leaves slack less than $W$. Indeed, not all
occurrences fit, so some occurrence is omitted. If the slack
were at least $W$, that occurrence could be added, strictly
increasing profit.

Consequently,
\[
    C-W<\wt(g),\wt(x^*)\le C,
    \qquad
    |\wt(x^*)-\wt(g)|<W.
\]
Thus an optimum lies among the exact-weight solutions covered by
\cref{lem:local-profile}. Compute that profile and return its best
value at a weight at most $C$, together with the reconstructed
witness.
\end{proof}

\paragraph{Using small profits instead of small weights.}

For the small-profit parameterization, we reverse the role of
weight and profit in the exact-value computation: for each
candidate profit, find the minimum weight that attains it.
The important point is that only $P$ consecutive candidate
profits need to be considered.

\begin{corollary}[Small profits]
\label{cor:knapsack-profit}
For positive integer profits at most $P$ and arbitrary positive
integer weights, bounded Knapsack can be solved deterministically,
with a witness, in $O(N+P^2\log^3(P+2))$ operations.
\end{corollary}

\begin{proof}
The all-fit case is immediate. Otherwise, let $g$ be the longest
feasible density prefix.

The fractional optimum takes $g$ and then a proper fraction of
the next occurrence. That occurrence has profit at most $P$,
so the fractional optimum has profit strictly less than
$\prf(g)+P$. Since $g$ is feasible and integral profits are
integers, the optimum integral profit belongs to
\[
    \{\prf(g),\prf(g)+1,\ldots,\prf(g)+P-1\}.
\]

We compute the minimum weight needed to attain each profit in
this interval. Give occurrence $i$ the transformed weight
$\widetilde w_i=p_i$ and transformed objective
$\widetilde p_i=-w_i$, retaining its multiplicity.
If the transformed weight is exactly $t$, maximizing the transformed
objective is precisely the same as minimizing original weight
subject to original profit $t$.

The original density order remains valid. Indeed,
\[
    \frac{\widetilde p_i}{\widetilde w_i}
    =-\frac{w_i}{p_i},
\]
and decreasing $p_i/w_i$ gives the same order as decreasing
$-w_i/p_i$. Thus the same $g$ is a density prefix in the
transformed instance, and all required transformed weights
lie within $P-1$ of its weight $\prf(g)$.

One call to \cref{lem:local-profile}, now with maximum weight
parameter $P$, computes the minimum original weight for every
candidate profit. Return the largest candidate whose minimum
weight is at most $C$, together with its witness. The running time is
$O(N+P^2\log^3(P+2))$.
\end{proof}

\section{Multiple-Choice Knapsack with Small Ranges}
\label{sec:multiple-choice}
There are $r$ nonempty classes of alternatives $(a,p)$, and a
solution chooses exactly one alternative from each class. Weights
are integers, profits are rational, and the objective is to maximize
total profit subject to a capacity bound. Let $N$ be the total
number of alternatives and let $w$ be the largest difference between
weights within one class. Subtracting each class's minimum weight
from its alternatives, and their sum from the capacity, puts all
weights in $[0,w]$ without changing the problem.

\begin{theorem}[Multiple-choice Knapsack]\label{thm:multiple-choice}
Put $k=\min\{r,2w\}$. For $0<\delta\le1/2$, multiple-choice
Knapsack can be solved using
\[
 O\!\left(N\log(N+2)+w^2k\log(k+2)
                          \log\frac{k+2}{\delta}\right)
\]
arithmetic and comparison operations. If alternatives are sorted
by weight within each class, the input term is $O(N)$.
The work bound holds deterministically.
Infeasibility is detected without error; otherwise the returned choice
is feasible and is optimal with probability at least $1-\delta$.
\end{theorem}

We first use the reduction of Pawlewicz~\cite{Pawlewicz26} to
find a reference choice $g$ and $O(wk)$ candidate changes. Some
optimum makes at most $k$ of these changes, and their signed
weights sum to a value in $[0,w)$. A sequential DP may nevertheless
need weight changes of order $kw$: it can encounter the positive
changes before the negative ones that cancel them.

We assign classes randomly to a binary tree so that cancellation
occurs within its subtrees. A concentration bound determines the
width to retain at each node. We then combine child tables by
enumerating pairs of states. We first give the reference choice
and candidate reduction, followed by this combination algorithm.

\subsection{A Reference Choice Optimal at Its Own Weight}
\label{sec:mc-anchor}
Keep only the most profitable alternative at each weight, then
discard any alternative whose profit is no larger than that of a
lighter one. Weights and profits now strictly increase within each
class. The lightest alternatives detect infeasibility. If the
heaviest alternatives fit, they form an optimum. These rules also
handle $w=0$; assume neither stopping rule applies.

Form the upper concave hull of each class's weight-profit points.
Moving from one hull vertex to the next increases weight by
$\Delta a$ and profit by $\Delta p$; regard this as an upgrade
of density $\Delta p/\Delta a$. These densities decrease along
each hull. Thus a global decreasing-density order, with ties
respecting each hull, never takes an upgrade before its predecessor
in the same class.

Start at the lightest alternative of every class, and take these
upgrades until the next edge would exceed capacity $c$. Let $g$
be the resulting choice and $c_g$ its weight. Each upgrade has
weight at most $w$, so
\begin{equation}\label{eq:mc-anchor-slack}
 0\le c-c_g<w.
\end{equation}
The important property is that $g$ is optimal at its own weight
$c_g$. In the fractional relaxation, one may take a fraction of
an upgrade. Decreasing density then gives an optimum just as in
fractional Knapsack. At capacity $c_g$, this procedure stops
exactly at the original alternatives selected by $g$. Every
alternative of every class lies on or below its hull, so no
integral choice of weight at most $c_g$ can have greater profit.

A stack scan builds the hulls from sorted input. Linear-time
selection on edge slopes, using weight increments as masses,
finds the stopping edge without sorting all edges globally.
This is the same weighted selection used for the density prefix
in \cref{sec:knapsack-kernel}. The preprocessing therefore takes
$O(N)$ operations after sorting. Alternatives below the hulls
are retained for the subsequent computation.

\subsection{Reducing to a Short List of Changes}
\label{sec:mc-candidates}
Represent an alternative in class $i$ by its changes from $g_i$:
\[
 d=a-a(g_i)\in[-w,w],\qquad v=p-p(g_i).
\]
The unchanged alternative has $(d,v)=(0,0)$. Optimality of $g$ at its own weight is the substitute for the ordered
exchange rule used in ordinary Knapsack. Any weight-zero collection
of changes in distinct classes has nonpositive profit gain from $g$. Undoing it in an
optimum therefore preserves optimality while using fewer changes.
This is exactly the obstruction needed by the walk lemma.

\begin{lemma}[Few changed classes]\label{lem:mc-proximity}
Some optimum at capacity $c$ either equals $g$ or changes at most
$2w-1$ classes. Its total weight change from $g$ belongs to $[0,w)$.
\end{lemma}
\begin{proof}
Choose an optimum $x$ with the fewest changed classes. If its
weight is at most $c_g$, then $g$ is also optimal and may be used.
Otherwise its total weight change $D$ satisfies $0<D<w$ by
\eqref{eq:mc-anchor-slack}.

Suppose a nonempty collection of its changes has total weight
zero. Applying just those changes to $g$ preserves weight, so
their total profit change is nonpositive. Undoing them in $x$
preserves feasibility, does not decrease profit, and reduces the
number of changed classes, a contradiction.
Distinct alternatives have distinct weights after preprocessing.
The nonzero changes thus form a zero-sum-free sequence in
$[-w,w]$, and \cref{lem:walk} bounds their number by $2w-1$.
\end{proof}

Now consider candidates with the same nonzero weight difference $d$.
An optimum changes at most $k=\min\{r,2w\}$ classes. If we keep
the $k$ candidates of largest profit change, then any omitted
candidate can be replaced by a kept one from an unused class.
Keep all candidates if there are fewer than $k$.
There is at most one difference-$d$ candidate per class, and at most
$M=2wk$ retained candidates in total. The following argument makes
the replacement precise.

\begin{lemma}[Keeping the best changes]\label{lem:mc-pruning}
Some optimum uses at most $k$ changed classes, only retained
changes, and a total weight change in $[0,w)$.
\end{lemma}
\begin{proof}
Fix the weight of an optimum from \cref{lem:mc-proximity}.
Among optima of this weight with at most $k$ changed classes,
choose one using as many retained changes as possible.
If it uses an omitted difference-$d$ change in class $i$, none
of the $k$ retained difference-$d$ changes belongs to $i$.
At most $k-1$ belong to other changed classes, so one belongs
to an unchanged class $j$.
Restore class $i$ to $g_i$ and make that retained change in $j$.
Weight is unchanged and profit does not decrease, but the number
of retained changes increases, a contradiction.
\end{proof}

Selection within the difference lists takes linear time.
Fix classes with no retained change at their reference alternatives.

\subsection{Narrow Tables in a Random Tree}\label{sec:mc-tree}
We will now assign the classes to a binary tree to exploit cancellation.
Fix one optimum from \cref{lem:mc-pruning}, before choosing the
random partition. It changes at most $k$ classes, with signed
weights $d_j\in[-w,w]$ of total $D\in[0,w)$. Assigning a class
to a random leaf keeps these changes independent across classes.

The number of leaves is chosen so that each leaf receives only
$O(\log(k/\delta))$ desired changes with high probability. This
bounds every partial sum while the leaf processes its classes.
For a larger subtree, we use concentration of the signed sum;
its table can be narrower than the sum of the leaf widths.
Set
\[
 L=\left\lceil\ln\frac{128(k+2)}{\delta}\right\rceil+1,
 \qquad
 b=\text{the smallest power of two at least }\max\{1,k/L\}.
\]
Assign each active class independently to a uniform leaf of a
complete binary tree with $b$ leaves. A subtree with $h$ leaves
receives each desired change with probability $h/b$. It therefore
has expected total weight change $(h/b)D<w$ and variance at most
$(h/b)kw^2\le hLw^2$. Bernstein's inequality bound gives size $O(wL\sqrt h)$ with failure probability
$2e^{-L}$. Accordingly, retain the interval
\begin{equation}\label{eq:mc-radii}
 [-R_h,R_h],\qquad R_h=10wL\lceil\sqrt h\rceil.
\end{equation}

At a leaf, process one class at a time, allowing a retained change
or the reference alternative. At an internal node, combine the
two child tables by max-plus convolution. \Cref{alg:mc-combine}
gives the complete computation; array entries outside their retained
intervals are $-\infty$.

\begin{algorithm}[H]
\caption{Combining the retained multiple-choice changes}
\label{alg:mc-combine}
\begin{algorithmic}[1]
\Require Retained changes $C_i=\{(d,v)\}$ for each active class;
$b$ leaves and radii $R_h$ from \eqref{eq:mc-radii}.
\Ensure A table of feasible weight and profit changes, with witnesses.
\State Assign each class independently to a uniform leaf of the $b$-leaf binary tree
\For{each leaf $\ell$}
  \State $F_\ell(0)\gets0$; all other entries are $-\infty$
  \For{each class $i$ assigned to $\ell$}
    \State For $|t|\le R_1$, set
      $F'_\ell(t)\gets\max_{(d,v)\in C_i\cup\{(0,0)\}}\{F_\ell(t-d)+v\}$
    \State $F_\ell\gets F'_\ell$
  \EndFor
\EndFor
\For{each internal node $u$ in bottom-up order}
  \State Let $a,c$ be its children and $h$ its number of descendant leaves
  \State Compute $F_u\gets(F_a\xc F_c)|_{[-R_h,R_h]}$ by enumerating all pairs of child states
\EndFor
\State \Return the root table, retaining maximizing transitions for backtracking
\end{algorithmic}
\end{algorithm}

These widths also give the desired work bound. A merge at an
$h$-leaf node enumerates $O(w^2L^2h)$ pairs, and a level has
$b/h$ such nodes. Thus every internal level costs
$O(bw^2L^2)=O(kw^2L)$ when $b>1$. The following lemma proves
that the retained states include the fixed optimum.

\begin{lemma}[The combination tree]\label{lem:mc-combination}
\Cref{alg:mc-combine} uses
$O(MwL+kw^2L\log(k+2))$ operations for every choice of the
randomness. Fix any choice of at most $k$ retained changes in
distinct classes, of total weight $D\in[0,w)$. With probability
at least $1-\delta$, the root entry at $D$ has profit at least
that of this choice. Every finite entry has a valid witness.
\end{lemma}
\begin{proof}
Assume first that $b>1$. At an $h$-leaf node, each desired change
is present independently with probability $\rho=h/b$, since the
changes belong to distinct classes. Its total contribution is
$S=\sum_jd_jI_j$, so
\[
 \mathbb E S=\rho D,\qquad
 \operatorname{Var}(S)\le\rho kw^2\le hLw^2.
\]
Bernstein's inequality gives
$|S-\rho D|\le wL(\sqrt{2h}+2/3)$ except with probability
$2e^{-L}$. Since $|\rho D|<w$, the contribution lies in the
retained interval.

At a leaf, the desired number of changed classes has expectation
at most $L$. Another application of Bernstein bounds it by $4L$
except with probability $2e^{-L}$. Every prefix of its changes
then has absolute weight at most $4wL$, so the sequential leaf
updates preserve it. A union bound over fewer than $2b$ nodes and the $b$ leaf
occupancies gives simultaneous survival with probability at least
$1-\delta$. If $b=1$, there are at most $k\le L$
desired changes, so every prefix survives deterministically.

The leaf updates have width $O(wL)$ and total cost $O(MwL)$,
since each retained alternative belongs to one class at one leaf;
the unchanged alternatives add at most the same order of work.
Initializing empty leaves costs $O(b)$, which is covered by the
claimed bound. If $b>1$, each internal level costs
\[
 \frac bh\,O(w^2L^2h)=O(bw^2L^2)=O(kw^2L),
\]
using $b<2k/L$. There are $O(\log(k+2))$ levels. If $b=1$
there are no internal merges.

All transitions combine disjoint classes and choose exactly one
alternative in each. Store maximizing alternatives and
pairs during the updates; backtracking is bounded by the table
construction time.
\end{proof}

\begin{proof}[Proof of \cref{thm:multiple-choice}]
Run \cref{alg:mc-combine} on the $M\le2wk$ retained changes. At the root,
maximize profit over entries of weight at most $c-c_g$, then
restore the reference profit and the fixed classes.
The unchanged choice is always available, so the returned solution
is feasible. On the survival event for the optimum from
\cref{lem:mc-pruning}, it is optimal. The tree cost and the
preprocessing give the stated bound.
\end{proof}

\subsection{One Coefficient of a Multiple Convolution}
The preceding algorithm also computes one prescribed coefficient
of a min-plus convolution of many sequences. The reduction makes
reaching the desired index the primary objective, without enlarging
the within-class ranges.

\begin{corollary}[One coefficient of a multiple min-plus convolution]
\label{cor:multiple-convolution}
Let $r$ sequences have $N$ finite rational entries in total, at
least one in each sequence. Suppose the finite-entry indices in
each sequence span an interval of length at most $w$.
One prescribed coefficient of their min-plus convolution can be
computed in $\softO(N+w^2\min\{r,w\})$ time with probability at
least $1-\delta$. The work bound holds deterministically. Every finite returned value has an original witness
and is an upper bound on the optimum; an infeasible coefficient
always receives $+\infty$.
\end{corollary}
\begin{proof}
Write the finite entries of sequence $i$ as $c_i(a)$, and set
\[
 C_0=\sum_i\bigl(\max_a c_i(a)-\min_a c_i(a)\bigr),
 \qquad M_0=C_0+1.
\]
Give alternative $a$ in class $i$ weight $a$ and profit
$M_0a-c_i(a)$. Run multiple-choice Knapsack at the requested
index $t$. If the capacity problem is infeasible, return $+\infty$.
Otherwise, if the returned weight is $t$, return the original cost
of this choice; if it is smaller, return $+\infty$.

The costs of any two complete choices differ by at most $C_0$.
Because the indices are integers, the modified objective first
maximizes weight at most $t$, then minimizes cost. If weight $t$
is attainable, an optimum therefore gives the desired coefficient.
The weight ranges are unchanged. Even on an unsuccessful run,
a finite returned value has a witness.
\end{proof}

\section{Output-Sensitive Subset Sum in Arbitrary Dimension}
\label{sec:subset-sum}
Let $X$ be a multiset of $n\ge1$ nonnegative integer vectors in
dimension $d$, and let $t\ge0$ be an integer target. We seek the
set $\SubSums{X,t}$ of all attainable sums in $[0,t]^d$ and write
$s=|\SubSums{X,t}|$. We improve the scaling of the algorithm of \cite{BFN25} with dimension.

\begin{theorem}[Output-sensitive Subset Sum]\label{thm:random-subset-sum}
For every fixed $d\ge1$ and $0<\delta\le1/2$, there is a randomized
algorithm with expected running time $\softO(n+s\sqrt n)$ that
returns a subset of $\SubSums{X,t}$. With probability at
least $1-\delta$, it returns the entire set within the same time
bound. More explicitly, the bound has the form
\begin{equation}\label{eq:ss-main-time}
 \operatorname{poly}(d)(n+s\sqrt n)
 \left(\log\frac{2(n+1)(d+1)(s+1)}{\delta}\right)^{d+O(1)}.
\end{equation}
We use the word-RAM assumptions of the sparse
nonnegative-convolution primitive; each vector occupies $O(d)$
words.
\end{theorem}

We follow the small/large item strategy of Bringmann, Fischer, and
Nakos~\cite{BFN25}. Small items are split among recursive calls.
A solution uses only a few large items, so color coding can place
them in distinct classes; we then insert at most one item per class.
Two computations determine the cost: merging the recursive answers
and inserting each color class. We bound the full recursive merge
by sumset submultiplicativity and compute color-class insertions
with a box-restricted sumset algorithm. The latter improves the
higher-dimensional running time.

\subsection{Box-Restricted Sumsets}
\label{sec:box-sumsets}
A full sumset can be much larger than its restriction to the target
box. To avoid generating these extra sums, we partition the valid input
pairs into products $A_i\times B_i$ such that every sum in
$A_i+B_i$ lies in the box. Full sparse convolution can then process
each product. The partition is useful because it copies each input
point into only a small number of products.

\begin{lemma}[Box-restricted sumsets]\label{lem:clipped-sumset}
Let $A,B\subseteq[0,t]^d\cap\Z^d$ and
$C=(A+B)\cap[0,t]^d$. Put $a=|A|$, $b=|B|$, and $c=|C|$.
The set $C$ can be computed in expected time
\begin{equation}\label{eq:box-time}
 \operatorname{poly}(d)(a+b+\sqrt{abc})
             \bigl(1+\log(a+b+2)\bigr)^{d+1}.
\end{equation}
\end{lemma}
\begin{proof}
Empty inputs are immediate. Sort $B$ separately in each coordinate,
breaking ties consistently. Let $r_j(y)\in\{0,\ldots,b-1\}$ be
the rank of $y\in B$ in the $j$th list. For $x\in A$, let
$q_j(x)$ count the points of $B$ with $j$th coordinate at most
$t-x_j$. Binary search computes these counts. A pair is valid
exactly when
\begin{equation}\label{eq:box-ranks}
 x+y\in[0,t]^d
 \quad\Longleftrightarrow\quad
 r_j(y)<q_j(x)\quad\text{for every }j.
\end{equation}
This remains true with repeated coordinate values, since the
count includes the entire tied block.

Use $h=\lceil\log(b+1)\rceil$ bits for every rank and count.
First consider one comparison $r<q$. At their first differing
bit, $r$ has $0$ and $q$ has $1$. If $p$ is the common prefix,
all numbers starting with $p0$ are smaller than all numbers starting
with $p1$. Thus one prefix identifies
a product of valid comparisons, and every strict comparison has
exactly one such prefix.

Apply this construction in each coordinate. A point $x\in A$
receives a key $(p_1,\ldots,p_d)$ for each choice of a prefix
immediately before a $1$-bit in every $q_j(x)$. A point $y\in B$
receives the analogous keys immediately before $0$-bits in its
ranks $r_j(y)$. Store each prefix by its length and value.
Group points with a common key into sets $A_i,B_i$, discarding
keys present on only one side. Each point receives at most $h^d$
keys, each valid pair occurs in exactly one product, and every
pair in such a product is valid. Consequently,
\[
 C=\bigcup_i(A_i+B_i),\qquad A_i+B_i\subseteq C,
 \qquad \sum_i|A_i|\le ah^d,\quad\sum_i|B_i|\le bh^d.
\]

Compute every $A_i+B_i$ by full sparse convolution.
Its size is at most $\min(c, |A_i||B_i|)$, so
\begin{align*}
 \sum_i|A_i+B_i|
 &\le\sqrt c\sum_i\sqrt{|A_i||B_i|}\\
 &\le\sqrt{c\Bigl(\sum_i|A_i|\Bigr)
                 \Bigl(\sum_i|B_i|\Bigr)}
 \le h^d\sqrt{abc}.
\end{align*}
This bounds the total number of emitted sums, including repetitions.
Sort the keys to construct the products, and
concatenate and sort the emitted lists to take their union.
The expected sparse-convolution cost and these sorting costs
give \eqref{eq:box-time}. In particular, the polynomial factor
in $d$ pays for comparisons of keys and vector encodings.
\end{proof}

\subsection{Controlling a Merge with Sumset Inequalities}
\label{sec:ss-sumset-inequalities}

The main difficulty in merging child answers is that the full sumset can
contain many sums outside the parent's box. An output-sensitive bound for
the parent is useful only if these intermediate sums are also controlled.

We arrange this by giving each of $k$ children target
$\lfloor u/(k-1)\rfloor$, where $u$ is the parent target. Any sum using
only $k-1$ children then fits in the parent box. The full $k$-child sum
need not fit, but the second inequality below bounds its size using these
leave-one-out sums. We reproduce the proof of \cite{GMR10} for completeness.

\begin{lemma}[Two sumset inequalities]\label{lem:ss-sumsets}
For nonempty finite $A_1,\ldots,A_k\subseteq\Z^d$, where $k\ge2$,
\begin{align}
    \sum_{i=1}^k(|A_i|-1)
    &\le |A_1+\cdots+A_k|-1,
    \label{eq:ss-superadditivity}\\
    |A_1+\cdots+A_k|^{k-1}
    &\le \prod_{i=1}^k\left|\sum_{j\ne i}A_j\right|.
    \label{eq:ss-submultiplicativity}
\end{align}
\end{lemma}

\begin{proof}
For \eqref{eq:ss-superadditivity}, sort each $A_i$ lexicographically.
Start with the sum of the least elements, then advance the first summand
through its list, the second through its list, and so on. Each advance
strictly increases the sum in lexicographic order. We obtain
$1+\sum_i(|A_i|-1)$ distinct sums.

For \eqref{eq:ss-submultiplicativity}, choose the lexicographically first
representing tuple for every element of $A_1+\cdots+A_k$, and let $V$ be
the set of these canonical tuples. Thus $|V|=|A_1+\cdots+A_k|$.
Write $\pi_{-i}V$ for the projection that deletes the $i$th entry.

A projected tuple is itself the lexicographically first representation
of its partial sum. Indeed, replacing it by a smaller representation
would preserve the full sum and produce a smaller full tuple, contradicting
canonicity. Consequently, two different tuples in $\pi_{-i}V$ have
different partial sums, and
\[
    |\pi_{-i}V|\le\left|\sum_{j\ne i}A_j\right|.
\]
It remains to relate the size of $V$ to the sizes of these projections.
We use the following form of Shearer's inequality.

Let $Z=(Z_1,\ldots,Z_k)$ be uniform on $V$. By the entropy chain rule
and the fact that conditioning reduces entropy,
\[
\begin{aligned}
    H(Z_{-i})
    &=\sum_{j\ne i}
      H\bigl(Z_j\mid (Z_\ell)_{\ell<j,\,\ell\ne i}\bigr)\\
    &\ge\sum_{j\ne i}H(Z_j\mid Z_1,\ldots,Z_{j-1}).
\end{aligned}
\]
On summing over $i$, each term of the full chain rule occurs exactly
$k-1$ times. Hence
\[
    (k-1)\log|V|
    =(k-1)H(Z)
    \le\sum_iH(Z_{-i})
    \le\sum_i\log|\pi_{-i}V|.
\]
Exponentiating and using the projection bounds proves
\eqref{eq:ss-submultiplicativity}.
\end{proof}

To apply
them here, partition the parent items into $k$ disjoint parts, and let
$T_i$ be the true subset-sum set of part $i$ with target
$\lfloor u/(k-1)\rfloor$. If the parent has $s$ attainable sums, then
\[
    \sum_{j\ne i}T_j\subseteq\SubSums{Y,u}
    \qquad\text{for every }i:
\]
each coordinate of such a sum is at most $u$, and the child item sets
are disjoint. Thus, writing $M=|T_1+\cdots+T_k|$,
\[
    M^{k-1}\le s^k,
    \qquad
    \sum_i(|T_i|-1)\le M-1.
\]
In particular, $M\le s^{1+1/(k-1)}$. For $k=\Omega(\log s)$,
the full merge is only a constant factor larger than the parent output.
We can therefore compute it without truncation and discard the
out-of-box sums afterward.

The denominator $k-1$, rather than $k$, serves two purposes. It makes
all leave-one-out sums feasible, which controls the merge size, while
giving each child slightly more than a $1/k$ share of the parent target.
That slack will let a random partition preserve a fixed solution.

\subsection{The Recursive Algorithm}
\label{sec:ss-recursion}

Initially, delete zero vectors and vectors outside $[0,t]^d$.
Zeros do not change the attainable sums, and an out-of-box vector cannot
belong to a feasible subset because all coordinates are nonnegative.
This costs $O(dn)$ time. If no vector remains, return $\{0\}$.
Fix an integer $k\ge32$ and put $h=dk^4$; both parameters remain fixed
throughout one recursive run. We will later choose $k$ without knowing
the output size.

At a call $(Y,u)$ with $|Y|>k$, we separate the items according to
whether random partitioning is safe. An item is \emph{small} if all
its coordinates are at most $u/k^4$, and \emph{large} otherwise.
Small items are sent to recursive children. Large items stay at the
current node and are inserted after the child answers have been merged.

\paragraph{Small items: balance the sizes and leave slack in the target.}

Arrange the small items in rows of $k$, padding the last row with zeros,
and independently permute each row. The $k$ columns form the child
instances, with the padding removed. Thus every child receives at most
$\lceil|Y|/k\rceil$ items, regardless of the random choices.

Now fix a feasible subset of $Y$. For one coordinate and one child,
each row contributes either the weight of one selected item or zero.
The row contributions are independent, bounded by $u/k^4$, and have
total expectation at most $u/k$. Their variance is at most
$(u/k^4)(u/k)=u^2/k^5$, whereas the slack in the child target is
\[
    \frac{u}{k-1}-\frac uk
    =\frac{u}{k(k-1)}
    \ge\frac{u}{k^2}.
\]
The slack is therefore at least $\sqrt{k}$ times the upper bound on
the standard deviation. The tail bound below gives failure probability
$e^{-\Omega(k)}$ for each coordinate of each child.

The row construction is useful for two separate reasons: it preserves
this concentration while also balancing the child input sizes
deterministically. Independent assignment of individual items
would give the former but not the latter.

\paragraph{Large items: a feasible subset uses few of them.}

Every large item exceeds $u/k^4$ in some coordinate. Assign each large
item of a fixed feasible subset to one such coordinate. Fewer than
$k^4$ can be assigned to any coordinate, so the subset contains fewer
than $h=dk^4$ large items in total.

Color all large occurrences independently with $2h^2$ colors. A fixed
set of fewer than $h$ occurrences has no collision with probability at
least $3/4$. When it has no collision, it can be recovered by processing
the color classes one at a time and choosing at most one occurrence
from each class. We repeat the coloring $k$ times and take the union
of the resulting answers.

Thus the two branches preserve a fixed solution in different ways:
small items are numerous but concentrate under partitioning; large
items need not concentrate, but there are few enough to separate by
coloring.

\begin{algorithm}[H]
\caption{$\textsc{Solve}(Y,u)$: recursive Subset Sum for a fixed $k$}
\label{alg:subset-sum}
\begin{algorithmic}[1]
\Require A multiset $Y$ of nonzero integer vectors in $[0,u]^d$;
fixed $k\ge32$ and $h=dk^4$.
\Ensure A subset of $\SubSums{Y,u}$ containing zero.
\If{$|Y|\le k$}
    \State $D\gets\{0\}$
    \For{each occurrence $y\in Y$}
        \State $D\gets(D\cup(D+y))\cap[0,u]^d$
    \EndFor
    \State \Return $D$
\EndIf
\State $Y_S\gets Y\cap[0,u/k^4]^d$;
       $Y_L\gets Y\setminus Y_S$
\State Arrange $Y_S$ in rows of $k$, padding the last row with zeros
\State Independently permute the items within each row
\State Let $Y_1,\ldots,Y_k$ be the columns with padding removed
\State $u'\gets\lfloor u/(k-1)\rfloor$
\For{$i=1,\ldots,k$}
    \State $A_i\gets\Call{Solve}{Y_i,u'}$
\EndFor
\State $A\gets A_1$
\For{$i=2,\ldots,k$}
    \State $A\gets A+A_i$ by full sparse convolution
    \Comment{Do not truncate yet.}
\EndFor
\State $A\gets A\cap[0,u]^d$
\If{$Y_L=\varnothing$}
    \State \Return $A$
\EndIf
\State $Q\gets\varnothing$
\For{$k$ independent coloring trials}
    \State Color every occurrence of $Y_L$ independently and uniformly
           with one of $2h^2$ colors
    \State $D\gets A$
    \For{each nonempty color class $B$}
        \State $D\gets(D+(B\cup\{0\}))\cap[0,u]^d$
        \Comment{Use \cref{lem:clipped-sumset}.}
    \EndFor
    \State $Q\gets Q\cup D$
\EndFor
\State \Return $Q$
\end{algorithmic}
\end{algorithm}

Each small item fits in the child box: its coordinates are integers
and at most $u/k^4\le u/(k-1)$. Empty children return $\{0\}$.
Within a color class, the update permits one item or none, so equal
vectors may be deduplicated there. Only nonempty color classes are
stored and processed.

\subsection{Correctness and the Total Work}
\label{sec:ss-analysis}

We analyze correctness and work separately. Correctness asks whether
one fixed witness survives the random choices; it holds for every
$k\ge32$. The work analysis instead uses the \emph{true} output
sizes, not the sizes recovered by the run, and requires
$k=\Omega(\log s)$. Keeping these statements separate will be
important when we choose $k$ from a budget.

\begin{lemma}[Preserving a specified sum]
\label{lem:ss-pointwise}
For every $k\ge32$, \cref{alg:subset-sum} returns only valid sums.
Any specified root sum is omitted with probability at most
\begin{equation}\label{eq:ss-pointwise}
    n^2(dk+1)e^{-k/3}.
\end{equation}
\end{lemma}

\begin{proof}
Every stored sum is valid. This follows inductively because the child
instances are disjoint, the large-item color classes are disjoint from
one another and from the children, and each update respects the available
occurrences. Truncation can discard a valid sum but cannot create an
invalid one. In particular, zero is always retained.

For the preservation claim, fix one witness for the specified root sum
before exposing any randomness. At a recursive call, its restricted
witness is determined by the ancestor partitions. Condition on those
choices and on this restricted witness fitting in the current box.
We bound the probability that the call fails to pass it on correctly.

\paragraph{The small-item partition.}

Fix a child and a coordinate. Its witness contribution is a sum of
independent row contributions in $[0,u/k^4]$, with expectation at most
$u/k$. The target $u/(k-1)$ is $(1+1/(k-1))$ times this upper bound
on the mean. A Chernoff bound therefore gives
\[
\begin{aligned}
    \Pr\!\left[\text{contribution}>\frac{u}{k-1}\right]
    &\le
    \exp\!\left(-\frac{u/k}{3(u/k^4)(k-1)^2}\right)\\
    &=\exp\!\left(-\frac{k^3}{3(k-1)^2}\right)
    \le e^{-k/3}.
\end{aligned}
\]
If there are no small items, there is no failure to bound. Otherwise
$u>0$, so the displayed application is valid. Contributions are
integers, so fitting below $u/(k-1)$ also means fitting below its floor.
A union bound over the $dk$ child-coordinate pairs gives failure
probability at most $dk e^{-k/3}$.

\paragraph{The large-item coloring.}

As observed above, the witness contains fewer than $h=dk^4$ large
occurrences. In one trial, a union bound over their pairs gives collision
probability at most $\binom h2/(2h^2)<1/4$. Thus the probability
that none of the $k$ independent trials separates them is at most
$4^{-k}$.

If the small-item witness sums survive the recursive calls and one
coloring separates the large-item witness, the call retains the desired
sum. Indeed, each color-class update can choose the required occurrence
or zero. Every intermediate sum is coordinatewise at most the final
witness sum, so none of these updates discards it.

The local failure probability is therefore at most
$dk e^{-k/3}+4^{-k}\le(dk+1)e^{-k/3}$.

\paragraph{Passing through the recursion.}

Base cases are exact. At an internal node with $n_v$ items, every
child has at most $\lceil n_v/k\rceil\le2n_v/k$ items. Hence the
depth is at most $\lceil\log_2 n\rceil$. Each non-root level has
at most $n$ calls, including empty ones: the parents have disjoint
inputs, and a parent creates $k$ children only when it contains more
than $k$ items. For $n\ge2$, the whole tree consequently has at
most $n^2$ calls; for $n\le1$, the root is exact.

If the root witness is lost, there is a first call at which either
its small-item partition or all its large-item colorings fail. Before
that call exposes its randomness, the restricted witness is already
fixed by its ancestors. The conditional local bound above therefore
applies at every possible call. A union bound over the calls proves
\eqref{eq:ss-pointwise}.
\end{proof}

\paragraph{The quantities used to charge the work.}

Let $n_v$ be the number of items at node $v$, and let $s_v$ be the
size of its true subset-sum set, regardless of what the randomized run
returns. Every such sum is also a valid root sum, so $s_v\le s$.
An empty call has $s_v=1$; a nonempty call has $s_v\ge2$, since
it contains a nonzero item that fits in its box.

\begin{lemma}[Work of one run]\label{lem:ss-core-time}
If $k\ge32+4\log s$, the expected work of
\cref{alg:subset-sum} is at most
\begin{equation}\label{eq:ss-core-time}
    \operatorname{poly}(d)k^5(n+s\sqrt n)
    \bigl(1+\log(n+s+2)\bigr)^{d+2}.
\end{equation}
The bound holds for every sequence of partitions and colorings.
\end{lemma}

\begin{proof}
We first bound the work performed at one node, excluding recursive
calls, and then sum it over the tree. Write
$\Lambda=1+\log(n+s+2)$.

\paragraph{Merging the small-item answers.}

At an internal node $v$, let $T_1,\ldots,T_k$ be the true child
output sets. Every leave-one-out sumset lies in the parent output.
Thus \cref{lem:ss-sumsets} gives
\begin{equation}\label{eq:ss-merge-size}
\begin{aligned}
    M_v:=|T_1+\cdots+T_k|
    &\le s_v^{k/(k-1)}\le\tfrac32s_v,\\
    \sum_i(|T_i|-1)
    &\le M_v-1\le2(s_v-1).
\end{aligned}
\end{equation}
Here $s_v^{1/(k-1)}\le e^{1/4}<3/2$ by the choice of $k$;
the last inequality uses $s_v\ge2$.

Every returned child answer contains zero and is contained in its true
set. Consequently, every intermediate full merge is contained in
$T_1+\cdots+T_k$: append zero from each child not yet merged.
The $k-1$ full sparse convolutions therefore cost $\softO(ks_v)$,
with a polynomial overhead in $d$.

\paragraph{Merging the large-item color classes.}

Suppose one color class has $m_j$ occurrences. The current set, the
next set, and the individual vectors in the class all belong to the
true parent output. Applying \cref{lem:clipped-sumset}, its insertion
costs at most
\[
    \operatorname{poly}(d)
    (s_v\sqrt{m_j}+m_j)\Lambda^{d+1}.
\]
There are at most $2h^2$ nonempty classes in a trial. Since their
total number of occurrences is at most $n_v$, Cauchy--Schwarz gives
\[
    \sum_j\sqrt{m_j}
    \le\sqrt{2h^2\sum_jm_j}
    \le\sqrt2\,h\sqrt{n_v}.
\]
Summing over the classes and then over the $k$ trials, and using
$h=dk^4$, bounds the local expected work by
\begin{equation}\label{eq:ss-local-cost}
    \operatorname{poly}(d)
    \bigl(kn_v+k^5s_v\sqrt{n_v}\bigr)\Lambda^{d+1}.
\end{equation}
This also covers the small-item merge, input scans, and unions.

At a nonempty leaf, keep the current set sorted. An update merges it
with a translated copy and discards out-of-box vectors in $O(ds_v)$
time. There are at most $k$ updates, so the leaf also satisfies
\eqref{eq:ss-local-cost}. Empty calls take constant work and are
counted separately below.

\paragraph{Relating recursive costs.}

We control the term $s_v\sqrt{n_v}$. The input size shrinks
rapidly at each level, but the output sizes of the children need not
sum to at most the parent's output size. Equation~\eqref{eq:ss-merge-size}
gives the weaker bound we need: after subtracting the common zero
output, their sizes sum to at most twice the parent's.

Set $\Phi_v=(s_v-1)\sqrt{n_v}$. Every child has at most $2n_v/k$
items, so
\[
\begin{aligned}
    \sum_{i=1}^k\Phi_i
    &\le\sqrt{\frac{2n_v}{k}}\sum_i(|T_i|-1)\\
    &\le2\sqrt{\frac2k}\,(s_v-1)\sqrt{n_v}
    \le\frac12\Phi_v,
\end{aligned}
\]
where the last inequality uses $k\ge32$.

Thus the possible factor-two growth in child output sizes is more than
offset by the decrease in the square root of the input size. The total
potential over each successive depth is at most half the previous one,
and hence
\[
    \sum_v\Phi_v\le2(s-1)\sqrt n.
\]
For nonempty calls, $s_v\le2(s_v-1)$. Therefore
\[
    \sum_{v:\,n_v>0}s_v\sqrt{n_v}
    \le2\sum_v\Phi_v
    \le4(s-1)\sqrt n.
\]
Subtracting one was essential: zero appears in every child, even an
empty one, and should not be charged as a separate recursive output.

Finally, the item sets at any depth are disjoint, so their sizes sum
to at most $n$. The depth bound gives
$\sum_v n_v=O(n\log(n+2))$, and the number of empty calls is bounded
by the same quantity. Summing \eqref{eq:ss-local-cost} and absorbing
one additional logarithm proves \eqref{eq:ss-core-time}.
\end{proof}

\subsection{Choosing the Parameter without Knowing the Output Size}
\label{sec:ss-budget}

The work bound requires $k=\Omega(\log s)$, but the algorithm does
not know $s$. We also choose to only assume an expected runtime bound for sparse nonnegative convolution. We handle
both issues by trying increasing budgets and restarting an attempt
that exceeds its budget.


Fix an error parameter $0<\delta<1$, and put
$R=\lceil\log_2(8/\delta)\rceil$. For each budget
$B=1,2,4,\ldots$, use
\begin{equation}\label{eq:ss-budget-parameter}
    k_B=\left\lceil C\log
    \frac{2(n+1)(d+1)R(B+1)}{\delta}\right\rceil,
\end{equation}
where $C$ is a sufficiently large absolute constant. This ensures
$k_B\ge32$. Make up to $R$ independent attempts with this parameter,
each capped at $B$ operations, and return the first completed answer.
The cost of writing the answer is included in the cap.

\begin{algorithm}[H]
\caption{Choosing the recursion parameter by budgets}
\label{alg:ss-budget}
\begin{algorithmic}[1]
\Require The preprocessed root instance $(Y,t)$ of at most $n$ occurrences,
and $0<\delta<1$.
\State $R\gets\lceil\log_2(8/\delta)\rceil$
\For{$B=1,2,4,\ldots$}
    \State Set $k\gets k_B$ from \eqref{eq:ss-budget-parameter}
           and $h\gets dk^4$
    \For{$j=1,\ldots,R$}
        \State Run $\Call{Solve}{Y,t}$ with fresh randomness,
               stopping the attempt after $B$ operations
        \Statex \hspace{\algorithmicindent}Count all work,
                including writing the output, toward this budget
        \If{the attempt completes within the budget}
            \State \Return its answer
        \EndIf
    \EndFor
\EndFor
\end{algorithmic}
\end{algorithm}

\begin{proof}[Proof of \cref{thm:random-subset-sum}]
We separately bound the chance of returning an incomplete answer and
the time needed for some attempt to finish. The initial preprocessing
cost is included in the final bound.

\paragraph{An incorrect answer is unlikely at every budget.}

For the parameter \eqref{eq:ss-budget-parameter},
\cref{lem:ss-pointwise} bounds the probability of omitting any fixed
true sum by
\[
    \frac{\delta}{16R(B+1)^3}.
\]
Interpret this guarantee for the complete, uncapped run using the
same random choices as the capped attempt. Whenever the attempt
finishes within its budget, its answer is exactly that of this complete
run.

If $s\le B$, an incorrect answer must omit one of the $s$ true
sums. A union bound over them gives error probability at most
$\delta/(16R(B+1)^2)$.

If $s>B$, fix any $B+1$ true sums before the attempt starts.
An attempt using at most $B$ operations cannot write all $B+1$ of
them. Completion therefore implies that one of these specified sums
was omitted. The same union bound again gives probability at most
$\delta/(16R(B+1)^2)$.

In both cases this bounds the event that an attempt completes
and returns an incorrect answer. Summing over the $R$ attempts at
each dyadic budget gives total error probability less than $\delta/4$.
Validity never fails.

\paragraph{A sufficiently large budget gives constant completion probability.}

Put $H=\log(2(n+1)(d+1)(s+1)/\delta)$. There is a dyadic budget
\[
    B_\star=\operatorname{poly}(d)(n+s\sqrt n)H^{d+O(1)}
\]
for which the expected work of one complete attempt is at most
$B_\star/2$. To see that the parameter choice is consistent with this
claim, note that $\log B_\star=O(dH)$. Hence $k_{B_\star}=O(dH)$,
and the budget can be chosen large enough that
$k_{B_\star}\ge32+4\log s$. Substituting into
\cref{lem:ss-core-time} gives a bound of the displayed form; enlarging
the absolute constant makes it at most $B_\star/2$.

The same conclusion holds at every larger dyadic budget $B$.
The only budget-dependent factor in \eqref{eq:ss-core-time} is
$k_B^5$, which grows as a fixed fifth power of a logarithm, whereas
the available time grows as $B$. Increasing $B_\star$ if necessary
makes the ratio at most $1/2$ and decreasing thereafter.

Markov's inequality now gives completion probability at least $1/2$
per attempt at these budgets. Thus the probability of exhausting all
$R$ attempts at $B_\star$ is at most
$2^{-R}\le\delta/8$. The total work through that budget is
$O(RB_\star)$, because the preceding budgets form a geometric series.
This has the form \eqref{eq:ss-main-time}. Together with the error
bound above, it gives the claimed joint correctness-and-time guarantee.

\paragraph{The expected time has the same bound.}

Let $q=2^{-R}\le\delta/8<1/8$. At every stage beyond $B_\star$,
conditional on reaching it, the probability of continuing to the next
stage is at most $q$. The budgets double, so the expected work after
$B_\star$ is bounded by
\[
    O(RB_\star)\sum_{j\ge1}(2q)^j
    =O(RB_\star).
\]
Adding the work through $B_\star$ proves the expected-time bound.
\end{proof}

\paragraph{The $s^{4/3}$ bound in one dimension.}

Bringmann and Nakos~\cite{BN21} also give a one-dimensional
restricted-sumset routine taking
$\softO(|A|+|B|+|C|^{4/3})$ time. Using it for the large-item
insertions recovers $\softO(n+s^{4/3})$ with the same recursion
and budget procedure. Correctness is unchanged; only the accounting
for the output-dependent work differs.

Previously, the factor $\sqrt{n_v}$ shrank enough to make the
recursive potential contract. The new local cost has no such factor,
so we instead bound the total output sizes directly at each depth.

At depth $j$, there are at most $k^j$ nodes. Their item sets are
disjoint and their targets are at most $t/(k-1)^j$. Any group of at
most $(k-1)^j$ nodes therefore has its full sumset contained in the
root output. By \eqref{eq:ss-superadditivity},
\[
    \sum_{v\text{ in the group}}(s_v-1)\le s-1.
\]
The budget parameter satisfies $k=\Omega(\log(n+2))$, and the depth
is $O(\log(n+2))$. Consequently,
\[
    \frac{k^j}{(k-1)^j}
    \le\exp\!\left(\frac{j}{k-1}\right)
    =O(1).
\]
Each depth can therefore be partitioned into only a constant number
of these groups, giving
\[
    \sum_{v\text{ at depth }j}(s_v-1)=O(s-1).
\]
On nonempty calls, $s_v\le2(s_v-1)$, so
\[
    \sum_{\substack{v\text{ at depth }j\\n_v>0}}s_v^{4/3}
    \le
    \left(\sum_{\substack{v\text{ at depth }j\\n_v>0}}s_v\right)^{4/3}
    =O(s^{4/3}).
\]

In one dimension, $h=k^4$, so all coloring trials together perform
$O(kh^2)=O(k^9)$ color-class insertions per node. The resulting
local work is $\softO(kn_v+k^9s_v^{4/3})$. Summing over depths
therefore gives expected work $\softO(k^9(n+s^{4/3}))$.
Its dependence on the budget is still a fixed power of
$\log B$, so the same stopping argument applies.

\section{Near-Convex Min-Plus Convolution}
\label{sec:near-convex}

Convexity makes min-plus convolution easy: merge the two sorted lists
of slopes. A small perturbation can destroy the order of the minimizing
pairs, but it cannot change their values by much. Rather than enumerate the pairs that might
minimize an output, it enumerates their distinct \emph{values}.

We use the geometric decomposition of Bringmann and Cassis~\cite{BC23}
to make this observation algorithmic. First discard pairs whose convex
reference value is too large for them to be optimal. Inside any
rectangle consisting entirely of the remaining pairs, each output
index has only $O(D+1)$ possible integer values. Sparse convolution
lists these values. The geometry supplies a cover by rectangles whose
total side length is $O(N\log(N+2))$.

\begin{theorem}[Near-convex convolution]
\label{thm:near-convex}
Let $f[0..n]$ and $g[0..m]$ be integer arrays. Suppose convex rational
arrays $F,G$ satisfy
\[
    0\le f_i-F_i\le\Delta_f,
    \qquad
    0\le g_j-G_j\le\Delta_g.
\]
Put $D=\Delta_f+\Delta_g$ and $N=n+m+2$. Then $f\mc g$ can be
computed in $\softO(N(D+1))$ time using output-sensitive nonnegative
convolution. The algorithm is deterministic when the sparse-convolution
primitive is deterministic.
\end{theorem}

We call a convex function $F$ such that $F_i \le f_i$ a convex minorant. A linear-time lower-hull
scan computes the greatest convex minorant of each array, 
interpolating between hull vertices. It lies above every
promised convex minorant, so replacing $F,G$ by these minorants cannot
increase the error bounds. A scan then finds the maximum gaps. We may
use their sum as $D$ throughout the algorithm.

\subsection{Relevant Pairs Have Few Distinct Values}
\label{sec:near-convex-tiles}

Let $H=F\mc G$. The nonnegative perturbations can only increase a
pair's value, while a pair attaining $H_k$ gains at most $D$. Hence
\begin{equation}\label{eq:nearconvex-value-range}
    H_k\le(f\mc g)_k\le H_k+D.
\end{equation}
Call a pair $(i,j)$ \emph{relevant} if
$F_i+G_j\le H_{i+j}+D$. Every minimizing pair is relevant: an
irrelevant pair already exceeds the upper bound in
\eqref{eq:nearconvex-value-range} before its perturbations are added.
For a relevant pair, adding the perturbations gives
\begin{equation}\label{eq:nearconvex-pair-range}
    H_{i+j}\le f_i+g_j\le H_{i+j}+2D.
\end{equation}

Now suppose every pair in a rectangle $I\times J$ is relevant, where
$I,J$ are integer intervals. Form the sumset of the two array graphs,
\[
    Z_{I,J}
    =\{(i,f_i):i\in I\}+\{(j,g_j):j\in J\}.
\]
Its first coordinate is an output index; its second is a candidate
value at that index. There are $|I|+|J|-1$ possible indices, and
\eqref{eq:nearconvex-pair-range} permits at most
$\lfloor2D\rfloor+1$ integer values at each one. Thus
\begin{equation}\label{eq:nearconvex-support}
    |Z_{I,J}|
    \le (|I|+|J|-1)(\lfloor2D\rfloor+1).
\end{equation}
This is the saving: the number of \emph{representations} may be
$|I||J|$, but the number of distinct graph sums is only
$O((|I|+|J|)(D+1))$. Nonnegative convolution counts representations
without listing them, so its output lists 
the candidate values we need. Retain the smallest one at each index.

To use a scalar sparse-convolution primitive, put
$f_{\min}=\min_{i\in I}f_i$ and $g_{\min}=\min_{j\in J}g_j$,
and choose an integer $B>n+m$. Encode the input points as
$i+B(f_i-f_{\min})$ and $j+B(g_j-g_{\min})$. Their sum is
\[
    (i+j)+B(f_i+g_j-f_{\min}-g_{\min}).
\]
Because $0\le i+j<B$, the residue modulo $B$ recovers the output
index and the quotient recovers the shifted value. The shifts make
all encoded inputs nonnegative, even when the array entries are
negative.

It remains to find relevant rectangles without scanning the full
index grid. The convex reference arrays give the
monotonicity needed for this step.

\subsection{Structure of the Relevant Pairs}
\label{sec:near-convex-band}

Write $\alpha_i=F_{i+1}-F_i$ and $\beta_j=G_{j+1}-G_j$ for the
two nondecreasing slope lists. A pair $(i,j)$ selects the first $i$
slopes of $F$ and the first $j$ slopes of $G$. Subject to $i+j=k$,
the cheapest selection consists of the $k$ smallest slopes in their
merged order. With a fixed rule for ties, this selection is a prefix
of each list.

If the merged slopes are $\gamma_0,\ldots,\gamma_{n+m-1}$, then
\[
    H_k=F_0+G_0+\sum_{\ell<k}\gamma_\ell.
\]
During the merge, record the number $i_k^\star$ of selected
$F$-slopes. This gives both $H_k$ and a minimizing pair
$(i_k^\star,k-i_k^\star)$ for every $k$ in $O(N)$ operations.
For rational arithmetic, evaluate
$H_k=F_{i_k^\star}+G_{k-i_k^\star}$ directly rather than adding
many fractions with different denominators.

Define the \emph{excess} $R(i,j)=F_i+G_j-H_{i+j}$. It is the
extra cost of the two chosen prefixes over the cheapest prefixes of
the same total length. Relevance is precisely $R(i,j)\le D$.
Two properties of these prefix choices determine the shape of the
relevant region: on a fixed diagonal their costs are convex, and
between consecutive diagonals a relevant choice can always be extended
or shortened by one slope without increasing its excess.

\begin{lemma}[Relevant band]\label{lem:relevant-band}
On each diagonal $i+j=k$, the relevant indices form a nonempty
interval $l_k\le i\le u_k$. For $0\le k<n+m$,
\[
    l_{k+1}-l_k\in\{0,1\},
    \qquad
    u_{k+1}-u_k\in\{0,1\}.
\]
Consequently, the pairs $(l_k,k-l_k)$ and $(u_k,k-u_k)$ trace
two monotone grid paths containing the relevant region between them.
\end{lemma}

\begin{proof}
On diagonal $k$, the forward difference of $F_i+G_{k-i}$ is
$\alpha_i-\beta_{k-i-1}$, which is nondecreasing in $i$.
Its sublevel set at $H_k+D$ is therefore an interval, and it is
nonempty because it contains a minimizer of $H_k$.

Consider any two prefixes selecting $k<n+m$ slopes. Their cheapest
next available slope has value at most $\gamma_k$: among the $k+1$
cheapest slopes, at least one has not been selected, and the next
available slope in its list is no more expensive. Adding the cheaper
next slope therefore costs at most $H_{k+1}-H_k$, so it does not
increase excess. This extends any relevant pair to a relevant pair
on diagonal $k+1$, with its $i$-coordinate unchanged or increased
by one.

Conversely, among any $k+1$ selected slopes, the most expensive has
value at least $\gamma_k$. It is the last selected slope of one
of the lists. Deleting it preserves the prefix condition and decreases
the cost by at least $H_{k+1}-H_k$. Thus any relevant pair on
diagonal $k+1$ can be shortened to a relevant pair on diagonal $k$,
with its $i$-coordinate unchanged or decreased by one.

Extending the leftmost relevant pair gives $l_{k+1}\le l_k+1$;
shortening the leftmost pair on the next diagonal gives
$l_k\le l_{k+1}$. Extending and shortening the rightmost pairs
similarly give $u_k\le u_{k+1}\le u_k+1$.
\end{proof}

The paths need not be computed. We will use their monotonicity to
classify an entire rectangle from two corners, evaluating excess only
at those corners.

\subsection{Covering the Band by Relevant Rectangles}
\label{sec:near-convex-rectangles}

We will recurse, start by embedding the index grid in a power-of-two square, let this be the only initial square.
Then for each current square, intersect
it with $[0,n]\times[0,m]$ and classify the resulting rectangle.
If every pair is relevant, process the rectangle by sparse convolution.
If no pair is relevant, discard it. Otherwise, split the square into
four equal children and continue. A singleton is always classified,
so this process covers every relevant pair.

The important point is that subdivision occurs only near the two
boundary paths. A large region lying wholly inside the band is handled
by one convolution, regardless of how many pairs it contains.
\Cref{fig:relevant-band} illustrates the decomposition.

\begin{figure}[t]
\centering
\begin{tikzpicture}[x=0.43cm,y=0.43cm,font=\small]
\fill[black!3] (0,0) rectangle (8,8);
\foreach \a/\b/\c/\d in {
 0/0/2/2,0/2/1/3,1/2/2/3,1/3/2/4,
 2/0/3/1,2/1/3/2,3/1/4/2,2/2/4/4,
 2/4/3/5,3/4/4/5,3/5/4/6,
 4/2/5/3,4/3/5/4,5/3/6/4,4/4/6/6,
 4/6/5/7,5/6/6/7,5/7/6/8,
 6/4/7/5,6/5/7/6,7/5/8/6,6/6/8/8}
 \filldraw[fill=blue!10,draw=blue!65!black,line width=0.45pt]
 (\a,\b) rectangle (\c,\d);
\foreach \i in {0,...,7} {
 \foreach \j in {0,...,7} {
  \pgfmathtruncatemacro{\relevant}{abs(\i-\j)<=2}
  \ifnum\relevant=1
    \fill[blue!65!black] (\i+0.5,\j+0.5) circle (0.85pt);
  \else
    \fill[black!25] (\i+0.5,\j+0.5) circle (0.85pt);
  \fi
 }
}
\draw[->] (0,-0.35) -- (8.5,-0.35) node[right] {$i$};
\draw[->] (-0.35,0) -- (-0.35,8.5) node[above] {$j$};
\node[anchor=north] at (0.5,-0.4) {$0$};
\node[anchor=north] at (7.5,-0.4) {$7$};
\node[anchor=east] at (-0.4,0.5) {$0$};
\node[anchor=east] at (-0.4,7.5) {$7$};
\end{tikzpicture}
\caption{The relevant pairs for $F_i=i^2$ and $G_j=j^2$,
$0\le i,j\le7$, at excess threshold $D=2$ are exactly those with
$|i-j|\le2$ (blue dots). The outlined dyadic squares partition these
pairs. Squares wholly inside the band are processed at once; only
squares meeting a boundary require further subdivision.}
\label{fig:relevant-band}
\end{figure}
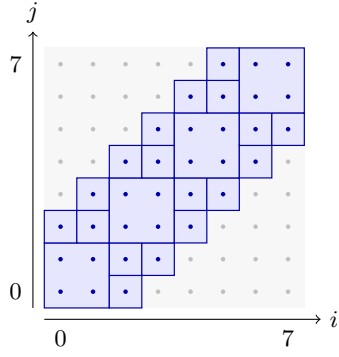

\paragraph{Why two corners suffice.}

For $I\times J=[a,b]\times[c,e]$, let $p=(a,e)$ and $q=(b,c)$.
To locate a pair relative to the band, consider
\[
    L(i,j)=i-l_{i+j},
    \qquad
    U(i,j)=i-u_{i+j}.
\]
Relevance means $L(i,j)\ge0$ and $U(i,j)\le0$.
By \cref{lem:relevant-band}, increasing $i$ by one increases both
quantities by zero or one; increasing $j$ by one decreases both by
zero or one. Hence both attain their minimum at $p$ and their
maximum at $q$.

We perform three tests. If $p$ and $q$ are both relevant,
then $L(p)\ge0$ and $U(q)\le0$ certify relevance throughout
the rectangle. If $q$ lies left of the band, then even the maximum
$L(q)$ is negative, so the whole rectangle lies left of it. If $p$
lies right of the band, then even the minimum $U(p)$ is positive,
so the whole rectangle lies right of it.

These tests require neither $l_k$ nor $u_k$. First test a corner
$(i,j)$ for relevance using $R(i,j)\le D$. If it is irrelevant,
compare $i$ with the known minimizing index $i_{i+j}^\star$:
it lies left of the relevant interval when $i<i_{i+j}^\star$,
and right when $i>i_{i+j}^\star$.

If none of the three tests applies, the rectangle contains both
relevant and irrelevant pairs. To see, walk from $p$ to $q$ by increasing $i$ or decreasing
$j$. Along this walk, $L$ and $U$ increase in steps of at most
one, with $L\ge U$. We have $U(p)\le0$ and $L(q)\ge0$. If
$p$ is not already relevant, then $L(p)<0$; the first point with
$L\ge0$ has $L=0$ and hence $U\le0$, so it is relevant. At
least one corner is nevertheless irrelevant, since otherwise the
acceptance test would apply. Thus every subdivided rectangle meets
a boundary of the band.

\begin{algorithm}[H]
\caption{Near-convex min-plus convolution}
\label{alg:near-convex}
\begin{algorithmic}[1]
\Require Integer arrays $f[0..n],g[0..m]$; convex minorants $F,G$
may be supplied.
\Ensure $C=f\mc g$.
\If{the minorants are not supplied}
    \State Compute the greatest convex minorants $F,G$ by lower-hull scans
\EndIf
\State $D\gets\max_i(f_i-F_i)+\max_j(g_j-G_j)$
\State Compute $H=F\mc G$ and minimizing indices $i_k^\star$
by merging slopes
\State Use $R(i,j)=F_i+G_j-H_{i+j}$ for constant-time excess queries
\If{$D=0$}
    \State \Return $H$
\EndIf
\State Set every $C_k\gets+\infty$
\State Let $Q$ be the smallest power of two at least $\max\{n+1,m+1\}$
\State Initialize a stack with the square $[0,Q-1]\times[0,Q-1]$
\While{the stack is nonempty}
    \State Pop a square $S$ and intersect it with $[0,n]\times[0,m]$
    \If{the intersection is empty}
        \State Continue to the next square
    \EndIf
    \State Write the intersection as $I\times J=[a,b]\times[c,e]$
    \State Set $p\gets(a,e)$ and $q\gets(b,c)$
    \If{$R(p)\le D$ and $R(q)\le D$}
        \State Compute $Z\gets\{(i,f_i):i\in I\}+\{(j,g_j):j\in J\}$
        \Statex \hspace{\algorithmicindent}using the graph encoding and
        sparse nonnegative convolution
        \For{each $(k,z)\in Z$}
            \State $C_k\gets\min\{C_k,z\}$
        \EndFor
    \ElsIf{$R(q)>D$ and $b<i_{b+c}^\star$}
        \State Discard $S$ \Comment{Every pair lies left of the band.}
    \ElsIf{$R(p)>D$ and $a>i_{a+e}^\star$}
        \State Discard $S$ \Comment{Every pair lies right of the band.}
    \Else
        \State Push the four children of $S$ onto the stack
    \EndIf
\EndWhile
\State \Return $C$
\end{algorithmic}
\end{algorithm}

\subsection{Correctness and Running Time}
\label{sec:near-convex-analysis}

\begin{proof}[Proof of \cref{thm:near-convex}]
Every minimizing pair is relevant. The rectangle decomposition discards
only irrelevant pairs, and every relevant pair eventually belongs to
an accepted rectangle. Sparse convolution enumerates all distinct
candidate values in each accepted rectangle. Taking their minimum at
each index therefore returns $f\mc g$.

For the work bound, fix a dyadic side length $\ell$. A monotone
boundary path meets $O(N/\ell+1)$ squares of that size: its two
coordinates never decrease, so it can cross each horizontal or vertical
grid line at most once. Every accepted square other than the root is
a child of a subdivided square, and every subdivided square meets a
boundary. Thus there are $O(N/\ell+1)$ accepted squares at scale
$\ell$, as well as this many squares requiring classification.
Clipping to a nonsquare input grid adds only the input grid's outer
boundaries to the count.

An accepted square contributes a rectangle with $|I|+|J|\le2\ell$.
By \eqref{eq:nearconvex-support}, its graph sumset has
$O(\ell(D+1))$ points. Since the root side length is $O(N)$,
the total support and input-scanning cost at this scale are
\[
    O\bigl((N/\ell+1)\ell(D+1)\bigr)=O(N(D+1)).
\]
There are $O(\log(N+2))$ scales. The sum of the support sizes over
all calls is therefore $O(N(D+1)\log(N+2))$. In particular,
every individual support has size $O(N(D+1))$, so an expected
$O(s\log(s+2))$ time sparse nonnegative convolution algorithm gives total expected time
\[
    O\bigl(N(D+1)\log(N+2)\log(N(D+1)+2)\bigr).
\]
The geometric tests take only $O(N)$ operations, since
$\sum_\ell(N/\ell+1)=O(N)$ over the dyadic scales. Using a
deterministic sparse primitive instead gives the stated deterministic
$\softO(N(D+1))$ bound.

When $D=0$, the input arrays equal their convex minorants, so slope
merging alone returns the answer in $O(N)$ operations.
\end{proof}

A minimizing pair for any one requested output can also be recovered
by scanning its diagonal in $O(N)$ time.

\begingroup
\raggedright
\bibliographystyle{alpha}
\bibliography{simpler_knapsack}
\endgroup
\end{document}